\documentclass[sigconf]{acmart}

\renewcommand\footnotetextcopyrightpermission[1]{}
\usepackage{balance} 
\usepackage[utf8]{inputenc}
\usepackage{amsmath}
\usepackage{amsfonts}
\usepackage{enumerate}
\usepackage{algorithm}
\usepackage[noend]{algpseudocode}
\usepackage{xspace}

\usepackage{tikz}
\usetikzlibrary{patterns}
\usetikzlibrary{calc}
\usepackage{tkz-tab}
\usepackage{caption}
\usepackage{latexsym}
\usepackage{amssymb}
\usepackage{subcaption}

\usepackage{paralist}

\newenvironment{cenumerate}
{\begin{compactenum}}
{\end{compactenum}}
\usepackage{fancyvrb}

\usepackage{comment}

\usepackage{graphicx}
\usepackage{pgfplots}
\pgfplotsset{compat=1.14}

\usepackage{cleveref}

\definecolor{light_blue}{HTML}{A6EEF7}
\definecolor{dark_blue}{HTML}{091EB8}
\definecolor{dark_red}{HTML}{EC0808}
\definecolor{light_red}{HTML}{F7A6B6}
\definecolor{light_purple}{HTML}{D192EA}
\definecolor{dark_purple}{HTML}{9114C2}
\definecolor{light_yellow}{HTML}{F8ED6D}
\definecolor{dark_yellow}{HTML}{FFF300}
\definecolor{dark_green}{HTML}{188120}
\definecolor{light_green}{HTML}{86DF8D}

\usepackage{float}
\usepackage{dblfloatfix}
\usepackage[T1]{fontenc}
\usepackage{caption}

\newcommand{\ENgeneral}[2]{\ensuremath{\mathsf{Enum}\langle #1, #2 \rangle}}
\newcommand{\RPgeneral}[2]{\ensuremath{\mathsf{REnum}\langle #1, #2 \rangle}}
\newcommand{\RAgeneral}[2]{\ensuremath{\mathsf{RAccess}\langle #1, #2 \rangle}}

\def\lin{\mathrm{lin}} 
\def\polylog{\mathrm{polylog}} 
\def\dconst{\mathrm{const}}
\newcommand{\EN}{\ENgeneral{\lin}{\log}}
\newcommand{\RP}{\RPgeneral{\lin}{\log}}
\newcommand{\RA}{\RAgeneral{\lin}{\log}}

\newcommand{\ENpolylog}{\ENgeneral{\lin}{\polylog}}
\newcommand{\RPpolylog}{\RPgeneral{\lin}{\polylog}}
\newcommand{\RApolylog}{\RAgeneral{\lin}{\polylog}}

\newcommand{\RPlogsquare}{\RPgeneral{\lin}{\log^2}}
\newcommand{\RAlogsquare}{\RAgeneral{\lin}{\log^2}}

\newcommand{\sparseBMM}{\textit{\e{\textsc{sparse-BMM}}}}
\newcommand{\Hyperclique}{\textit{\e{\textsc{Hyperclique}}}}
\newcommand{\Triangle}{\textit{\e{\textsc{Triangle}}}}

\newtheorem{theorem}{Theorem}[section]
\newtheorem{proposition}[theorem]{Proposition}
\newtheorem{lemma}[theorem]{Lemma}
\newtheorem{corollary}[theorem]{Corollary}
\newtheorem{definition}[theorem]{Definition}
\newtheorem{example}[theorem]{Example}
\newtheorem{fact}[theorem]{Fact}

\newenvironment{repeatresult}[2]
{\vskip0.5em\nobreakspace{\textsc{#1} #2.}\em}
{\vskip1em}

\newenvironment{reptheorem}[1]{\begin{repeatresult}{Theorem}{#1}}{\end{repeatresult}}

\def \wgt{\mathrm{w}}
\def \rng{\mathrm{startIndex}}
\newcommand \bucket[2]{\mathrm{bucket}[{#1},{#2}]} %
\newcommand \atts[1]{\mathrm{Atts_{#1}}}
\newcommand \patts[1]{\mathrm{pAtts_{#1}}}

\def \split{\textsc{SplitIndex}}
\def \combine{\algname{CombineIndex}}

\def \ans{\mathit{answer}}
\def \element{\mathit{element}}
\def \providers{\mathit{providers}}
\def \owner{\mathit{owner}}
\def \chosen{\mathit{chosen}}
\def \root{\mathit{root}}
\def \offset{\mathit{offset}}

\def\cqa{\,\mathbin{\mbox{:-}}\,}

\def\vars{\mathit{Vars}}
\def\H{\mathcal{H}}
\def \true{\mathit{True}}

\newcommand \zhaoalg[1]{\textsc{Sample}(#1)\xspace}
\newcommand \ouralgCQ{\textsc{REnum}(CQ)\xspace}
\newcommand \ouralgUCQ{\textsc{REnum(UCQ)}\xspace}
\newcommand \mcUCQAlg{\textsc{REnum(mcUCQ)}\xspace}

\newcommand{\set}[1]{\ensuremath{\{ #1 \}}}
\newcommand{\deff}{\ensuremath{:=}}
\newcommand{\EOE}{\texttt{EOE}}
\newcommand{\ErrorMessage}{\texttt{Error}}
\newcommand{\Access}{\algname{Access}}
\newcommand{\Rank}{\algname{InvAcc}}
\newcommand{\Largest}{\algname{Largest}}
\newcommand{\Algo}[1]{\ensuremath{\mathcal{#1}}}
\newcommand{\NNpos}{\ensuremath{\mathbb{N}_{\geq 1}}}
\newcommand{\RRpos}{\ensuremath{\mathbb{R}_{\geq 0}}}

\begin{document}
\pagestyle{plain}
\title{Answering (Unions of) Conjunctive Queries using Random Access and Random-Order Enumeration}

\author{Nofar Carmeli}
\affiliation{
  \institution{Technion}
  }

\author{Shai Zeevi}
\affiliation{
\institution{Technion}
  }

\author{Christoph Berkholz}
\affiliation{
\institution{Humboldt-Universit\"{a}t zu Berlin}
  }

\author{Benny Kimelfeld}
\affiliation{
\institution{Technion}
  }

\author{Nicole Schweikardt}
\affiliation{
\institution{Humboldt-Universit\"{a}t zu Berlin}
  }

\def\e#1{\emph{#1}}
\def\scs{\mathbf{S}}
\def\ar{\mathrm{ar}}
\def\const{\mathsf{Const}}
\newcommand{\eat}[1]{}
\def\tup#1{\vec{#1}}

\begin{abstract}
  As data analytics becomes more crucial to digital systems, so grows
  the importance of characterizing the database queries that admit a
  more efficient evaluation. We consider the tractability yardstick of
  answer enumeration with a polylogarithmic delay after a linear-time
  preprocessing phase. Such an evaluation is obtained by constructing,
  in the preprocessing phase, a data structure that supports
  polylogarithmic-delay enumeration. In this paper, we seek a
  structure that supports the more demanding task of a ``random
  permutation'': polylogarithmic-delay enumeration in truly random
  order. Enumeration of this kind is required if downstream
  applications assume that the intermediate results are representative
  of the whole result set in a statistically valuable manner. An even
  more demanding task is that of a ``random access'':
  polylogarithmic-time retrieval of an answer whose position is given.

  We establish that the free-connex acyclic CQs are tractable in all
  three senses: enumeration, random-order enumeration, and random
  access; and in the absence of self-joins, it follows from past
  results that every other CQ is intractable by each of the three
  (under some fine-grained complexity assumptions). However, the three
  yardsticks are separated in the case of a union of CQs (UCQ): while
  a union of free-connex acyclic CQs has a tractable enumeration, it
  may (provably) admit no random access. For such UCQs we devise a
  random-order enumeration whose delay is logarithmic in
  expectation. We also identify a subclass of UCQs for which we can
  provide random access with polylogarithmic access time. Finally, we
  present an implementation and an empirical study that show a
  considerable practical superiority of our random-order enumeration
  approach over state-of-the-art alternatives.
\end{abstract}

\maketitle

\def\algname#1{\textit{\e{\textsc{#1}}}}

\vspace{-1.2em}

\section{Introduction}\label{sec:intro}

In the effort of reducing the computational cost of database queries
to the very least possible, recent years have seen a substantial
progress in understanding the fine-grained data complexity of
enumerating the query answers. The seminal work of Bagan, Durand, and
Grandjean~\cite{bdg:dichotomy} has established that the
\e{free-connex acyclic} conjunctive queries (or just \e{free-connex
  CQs} for short) can be evaluated using an enumeration algorithm with
a constant delay between consecutive answers, after a linear-time
preprocessing phase. Moreover, their work, combined with that of
Brault-Baron~\cite{bb:thesis}, established that, in the absence of
self-joins (i.e., when every relation occurs at most once), the
free-connex CQs are \e{precisely} the CQs that have such an
evaluation. The lower-bound part of this dichotomy requires
some lower-bound assumptions in fine-grained complexity
(namely, that
neither \e{sparse Boolean matrix-multiplication},
nor \e{triangle detection}, nor \e{hyperclique detection} can be done
in linear time).
Later
generalizations consider unions of CQs
(UCQs)~\cite{nofar:ucq,DBLP:conf/icdt/BerkholzKS18} and the presence
of constraints~\cite{nofar:fds,DBLP:conf/icdt/BerkholzKS18}.

As a query-evaluation paradigm, the enumeration approach has the
important guarantee that the number of intermediate results is
proportional to the elapsed processing time. This guarantee is useful
when the query is a part of a larger analytics pipeline where the
answers are fed into downstream processing such as machine learning,
summarisation, and search. The intermediate results can be used to
save time by invoking the next-step processing (e.g., as in streaming
learning algorithms~\cite{DBLP:conf/kdd/StreetK01}), computing
approximate summaries that improve in time (e.g., as in online
aggregation~\cite{DBLP:conf/sigmod/HaasH99,
  DBLP:journals/tods/LiWYZ19}), and presenting the first pages of
search results (e.g., as in keyword search over structured
data~\cite{DBLP:conf/vldb/HristidisP02,DBLP:conf/sigmod/GolenbergKS08}).
Yet, at least the latter two applications make the implicit assumption
that the collection of intermediate results is a representative of the
entire space of answers. In contrast, the aforementioned
constant-delay algorithms enumerate in an order that is a merely an
artifact of the tree selected to utilize free-connexity, and hence,
intermediate answers may feature an extreme bias. Importantly, there
has been a considerable recent progress in understanding the ability
to enumerate the answers not just efficiently, but also with a
guarantee on the
order~\cite{DBLP:journals/corr/abs-1902-02698,DBLP:journals/corr/abs-1911-05582}.

Yet, to be a statistically meaningful representation of the space of
answers, the enumeration order needs to be provably random. In this
paper, we investigate the task of enumerating answers in a uniformly
random order. To be more precise, the goal is to enumerate the answers
without repetitions, and the output induces a uniform distribution
over the space of permutations over the answer set.  We refer to this
task as \e{random permutation}.  Similarly to the recent work on
ranked
enumeration~\cite{DBLP:journals/corr/abs-1902-02698,DBLP:journals/corr/abs-1911-05582},
our focus here is on achieving a \e{logarithmic} or \e{polylogarithmic
  delay} after a linear preprocessing time. Hence, more technically,
the goal we seek is to construct in linear-time a data structure that
allows to sample query answers \e{without replacement}, with a
(poly)logarithmic-time per sample. Note that sampling \e{with}
replacement has been studied in the
past~\cite{DBLP:conf/sigmod/AcharyaGPR99a,DBLP:conf/sigmod/ChaudhuriMN99}
and recently gained a renewed attention~\cite{zhao2018random}.

One way of achieving a random permutation is via \e{random access}---a
structure that is tied to some enumeration order and, given a position
$i$, returns the $i$th answer in the order. To satisfy our target of
an efficient permutation, we need a random-access structure that can
be constructed in linear time (preprocessing) and supports answer
retrieval (given $i$) in polylogarithmic time. We show that, having
this structure at hand, we can use the Fisher-Yates
shuffle~\cite{fisheryates:shuffle} to design a random permutation with
a negligible additive overhead over the preprocessing and enumeration
phases.

So far, we have mentioned three tasks of an increasing demand: \e{(a)}
enumeration, \e{(b)} random permutation, and \e{(c)} random access.
We show that all three tasks can be performed efficiently (i.e.,
linear preprocessing time and evaluation with polylogarithmic time per
answer) over the class of free-connex CQs. We conclude that within the
class of CQs without self-joins, it is the same precise set of queries
where these tasks are tractable---the free-connex CQs.  (We remind the reader that all mentioned lower bounds are under
assumptions in fine-grained complexity.) The existence
of a random access for free-connex CQs has been established by
Brault-Baron~\cite{bb:thesis}. Here, we devise our own random-access
algorithm for free-connex CQs that is simpler and better lends itself
to a practical implementation. Moreover, we design our algorithm in
such a way that it is accompanied by an \e{inverted access} that is
needed for our later results on UCQs.

The tractability of enumeration generalizes from free-connex CQs to
\e{unions} of free-connex
CQs~\cite{nofar:ucq,DBLP:conf/icdt/BerkholzKS18}. Interestingly, this
is no longer the case for random access! The reason is as follows. An
efficient random access allows to count the answers; while counting
can be done in linear time for free-connex CQs, we show the existence
of a union of free-connex CQs where linear-time counting can be used
for linear-time triangle detection in a graph. At this point, we are
investigating two questions:
\begin{cenumerate}
  \item Can we identify a nontrivial class of UCQs with an efficient
    random access?
  \item Can we get an efficient random permutation for unions of
    free-connex CQs, without requiring a random access?
    \end{cenumerate}
    For the first question, we 
    identify the class of
    \e{mutually-compatible UCQs} (mc-UCQs) and show that every such
    UCQ has an efficient random access. As for the second question, we
    show that the answer is positive under the following weakening of
    the delay guarantee: there is a random permutation where \e{each
      delay} is a geometric random variable with a logarithmic
    mean. In particular, each delay is logarithmic in expectation.

    Finally, we present an implementation of our random-access and
    random-permutation algorithms, and present an empirical
    evaluation. Over the TPC-H benchmark, we compare our random
    permutation to the approach of using a state-of-the-art random
    sampler~\cite{zhao2018random}, which is designed to produce a
    uniform sample with replacement, and then remove duplicates as
    they are detected. The experiments show that our algorithms are
    not only featuring complexity and statistical guarantees, but also
    a significant practical improvement. Moreover, the experiments
    show that Fisher-Yates over our random access for mc-UCQs can
    further accelerate the union enumeration (in addition to the
    deterministic guarantee on the delay), compared to our generic
    algorithm for UCQ random permutation; yet, this acceleration is
    not consistently evident in the experiments.
   
The paper is structured as
follows. The basic notation is fixed in
Section~\ref{sec:preliminaries}. In Section~\ref{sec:classes} we
introduce three classes of enumeration problems and discuss the
relationship between them.
Sections~\ref{sec:cqs} and \ref{sec:ucqs} are devoted to our results
concerning CQs and UCQs, respectively. Section~\ref{sec:experiments}
presents our experimental study.
Due to space restrictions, some details had to be deferred to an appendix.

\section{Preliminaries}\label{sec:preliminaries}

In this section, we provide basic definitions and notation that we
will use throughout this paper.
For integers $\ell,m$ we write $[\ell,m]$ for the set of all integers
$i$ with $\ell\leq i\leq m$.

\subsubsection*{Databases and Queries}
A (relational) \e{schema} $\scs$ is a collection of \e{relation
  symbols} $R$, each with an associated arity $\ar(R)$. A \e{relation}
$r$ is a set of tuples of \e{constants}, where each tuple has the same
arity (length). A \e{database} $D$ (over
the schema $\scs$) associates with each relation symbol $R$ a finite
relation $r$, which we denote by $R^D$, such that
$\ar(R)=\ar(R^D)$. Notationally, we identify a database $D$ with its
finite set of \e{facts} $R(c_1,\dots,c_k)$, stating that the relation
$R^D$ over the $k$-ary relation symbol $R$ contains the tuple
$(c_1,\dots,c_k)$.

A \e{conjunctive query} (CQ) over the schema $\scs$ is a relational
query $Q$ defined by a first-order formula of the form
$\exists {\tup y}\,\varphi(\tup{x}, \tup{y})$, where $\varphi$ is a
conjunction of atomic formulas of the form $R(\tup t)$ with variables
among those in $\tup{x}$ and $\tup {y}$.  We write a CQ $Q$ shortly as
a logic rule, that is, an expression of the form
$
Q(\tup x) \cqa R_1(\tup{t}_1),\dots,R_n(\tup{t}_n)
$
where each $R_i$ is a relation symbol of $\scs$, each $\tup{t}_i$ is a
tuple of variables and constants with the same arity as $R_i$, and
$\tup{x}$ is a tuple of $k$ variables from
$\tup{t}_1,\dots,\tup{t}_n$. We call $Q(\tup x)$ the \e{head} of $Q$,
and $R_1(\tup{t}_1),\dots,R_n(\tup{t}_n)$ the \e{body} of $Q$. Each
$R_i(\tup t_i)$ is an \e{atom} of $Q$. We use $\vars(Q)$ and
$\vars(\alpha)$ to denote the sets of variables that occur in the CQ
$Q$ and the atom $\alpha$, respectively. The variables occurring in
the head are called the \e{head variables}, and we make the standard
safety assumption that every head variable occurs at least once in the
body. The variables occurring in the body but not in the head are
existentially quantified, and are called the \e{existential
  variables}. 
A CQ with no existential variables is called a \e{full join query}.

We usually omit the explicit specification of the schema $\scs$, and
simply assume that it is the one that consists of the relation symbols
that occur in the query at hand.

A \e{homomorphism} from a CQ $Q$ to a database
$D$ is a mapping of the variables in $Q$ to the constants of $D$, such
that every atom of $Q$ is mapped to a fact of $D$. 
Each such homomorphism $h$ yields an \emph{answer} to $Q$, which is
obtained from $\tup x$ by replacing every variable in $\tup x$ with
the constant it is mapped to by $h$.
We denote by $Q(D)$ the set of all answers to $Q$ on $D$.

We say that a database $D$ is \e{globally consistent} with respect to $Q$ if each fact in $D$ \emph{agrees} with some answer in $Q(D)$; that is, there exists a homomorphism from $Q$ to $D$ and an atom of $Q$ such that the homomorphism maps the atom to the fact.

A \e{self-join} in a CQ $Q$ is a pair of distinct atoms over the same
relation symbol. We say that $Q$ is \e{self-join free} if it has no
self-joins, that is, every relation symbol occurs at most once in the
body.

To each CQ
$Q(\tup x) \cqa \alpha_1(\tup x, \tup y),\dots,\alpha_k(\tup x, \bar
y)$ we associate a hypergraph $\H_Q$ where the nodes are the variables
in $\vars(Q)$, and the edges are $E = \{e_1, ..., e_k\}$ such that
$e_i = \vars(\alpha_i)$. Hence, the nodes of $\H_Q$ are
$\tup x \cup \tup y$, and the hyperedge $e_i$ includes all the
variables that appear in $\alpha_i$. 
A CQ $Q$ is \emph{acyclic} if its hypergraph is acyclic.
That is, there exists a tree $T$ (called a \emph{join-tree} of $Q$) 
such that
$nodes(T) = edges(\mathcal{H}_Q)$, and for every
$v \in nodes(\mathcal{H}_Q)$, the nodes of $T$ that contain $v$ form a
(connected) subtree of $T$. The CQ $Q$ is \emph{free-connex} if $Q$ is
acyclic and $\mathcal{H}_Q$ remains acyclic when adding a hyperedge
that consists of the free variables of $Q$.

A \emph{union of CQs} (UCQ) is a query of the form
$Q_1(\tup x)\cup\dots\cup Q_m(\tup x)$, where every $Q_i$ is a
CQ with the sequence $\tup x$ of head variables. The set of answers to
$Q_1(\tup x)\cup\dots\cup Q_m(\tup x)$ over a database $D$ is,
naturally, the union $Q_1(D)\cup\dots\cup Q_m(D)$.

\subsubsection*{Computation Model}
An \emph{enumeration problem} $P$ is a collection of pairs $(I,Y)$
where $I$ is an \emph{input} and $Y$ is a finite set of \emph{answers}
for $I$, denoted by $P(I)$. An \emph{enumeration algorithm}
$\mathcal{A}$ for an enumeration problem $P$ is an algorithm that
consists of two phases: \e{preprocessing} and \e{enumeration}. During
preprocessing, $\mathcal{A}$ is given an input \emph{I}, and it
builds certain data structures. During the enumeration phase,
$\mathcal{A}$ can access the
data structures built during
preprocessing, and
it emits the answers $P(I)$, one by one, without repetitions. We
denote the running time of the preprocessing phase by $t_p$. The time
between printing any two answers during the enumeration phase is
called \emph{delay}, and is denoted by $t_d$. 

In this paper, an enumeration problem will refer to a query (namely, a CQ or
a UCQ) $Q$, the input $I$ is a database $D$, and the answer set $Y$ is
$Q(D)$. Hence, we adopt \emph{data
  complexity}, where the query is treated as fixed. We use a variant
of the \emph{Random Access Machine} (RAM) model with uniform cost
measure named DRAM~\cite{DBLP:journals/amai/Grandjean96}.
This model enables the construction of lookup tables of polynomial size
that can be queried in constant time.

\subsubsection*{Complexity Hypotheses}
Our conditional optimality results rely on the following hypotheses on
the hardness of algorithmic problems.

The hypothesis 
$\sparseBMM$ states that there is no algorithm that
multiplies two Boolean matrices (represented as lists of
their non-zero entries) over the Boolean semiring
in time $m^{1+o(1)}$, where $m$ is the number of non-zero entries in
$A$, $B$, and $AB$. The best known running time for this problem is
$O(m^{4/3})$ \cite{DBLP:conf/icdt/AmossenP09}, which remains true even if
the matrix multiplication exponent\footnote{The matrix multiplication
  exponent $\omega$ is the smallest number such that for any $\varepsilon
  > 0$ there is an algorithm that multiplies two rational $n\times n$
  matrices with at most $O(n^{\omega+\varepsilon})$
  (arithmetic) operations.
The currently best bound on $\omega$
is $\omega <
2.373$ and it is conjectured that $\omega=2$ \cite{DBLP:conf/stoc/Williams12, DBLP:conf/issac/Gall14a}.} $\omega$ is equal to $2$.

By $\Triangle$ we denote the hypothesis that there is no
$O(m)$ time algorithm
that detects whether a graph with $m$ edges contains a triangle. The
best known algorithm for this problem runs in time  $m^{2\omega/(\omega + 1) + o(1)}$ \cite{Alon.1997}, which is
$\Omega(m^{4/3})$ even if $\omega=2$. 
The $\Triangle$ hypothesis is also implied by a slightly stronger
conjecture in \cite{DBLP:conf/focs/AbboudW14}. 

A $(k{+}1,k)$-hyperclique is a set of $k{+}1$ vertices in a
hypergraph such that every $k$-element subset is a hyperedge.
By $\Hyperclique$  we denote the hypothesis that for every $k\geq 3$ there is no time
$O(m)$ algorithm for deciding the existence of a
$(k{+}1,k)$-hyperclique in an $k$-uniform hypergraph with $m$ hyperedges. 
This hypothesis is implied by the $(l,k)$-\textsc{Hyperclique} conjecture proposed in \cite{hypotheses}.

While the three hypotheses are not as established as classical complexity 
assumptions (like P\,$\neq$\,NP), their refutation would lead to 
unexpected breakthroughs in algorithms, which would be achieved when
improving the relevant methods in our paper.

\section{Enumeration Classes}\label{sec:classes}
In this section, we define three classes of enumeration problems and
discuss the relationship between them.

\subsection{Definitions}

We write $d$ to denote a function from the positive integers $\NNpos$ to
the non-negative reals $\RRpos$, and $d=\dconst$, $d=\lin$, $d=\log^c$ (for $c\geq 1$) mean
$d(n)=1$, $d(n)=n$, $d(n)=\log^c(n)$, respectively.

\begin{definition} 
Let $d$ be a function from $\NNpos$ to $\RRpos$.
We define
$\ENgeneral{\lin}{d}$ to be the class of enumeration problems for
  which there exists an enumeration algorithm $\mathcal{A}$ such that
  for every input $I$ it holds that $t_p \in O(|I|)$ and
  $t_d \in O(d(|I|))$. 
Furthermore, $\ENpolylog$ is the union of $\ENgeneral{\lin}{\log^c}$ for
all $c\geq 1$.
\end{definition}

A \emph{random-permutation algorithm} for an enumeration problem $P$
is an enumeration algorithm where every emission is done uniformly at
random. That is, at every emission, every until then not yet emitted tuple has
equal probability of being emitted. As a result, if $|P(I)|=n$, every
ordering of the answers $P(I)$ has probability $\frac{1}{n!}$ of
representing the order in which $\mathcal{A}$ prints the answers. 

\begin{definition} 
Let $d$ be a function from $\NNpos$ to $\RRpos$.
We define
$\RPgeneral{\lin}{d}$ to be the class of enumeration problems for
  which there exists a random-permutation algorithm $\mathcal{A}$ such that
  for every input $I$ it holds that $t_p \in O(|I|)$ and
  $t_d \in O(d(|I|))$. 
Furthermore, $\RPpolylog$ is the union of $\RPgeneral{\lin}{\log^c}$ for
all $c\geq 1$.
\end{definition}

\begin{fact}
By definition,
$\RPgeneral{\lin}{d}\subseteq\ENgeneral{\lin}{d}$ for all $d$.
\end{fact}

A \emph{random-access algorithm} for an enumeration problem $P$ is an
algorithm $\mathcal{A}$ consisting of a preprocessing phase and an access routine. The preprocessing phase builds a data structure based
on the input $I$. Afterwards, the access routine may be called any
number of times, and it may use the data structure built during
preprocessing. There exists an order of $P(I)$, denoted $t_1,...,t_n$
and called \emph{the enumeration order of $\mathcal{A}$}
such that, when the access routine is called with parameter $i$,
it returns $t_i$ if $1 \leq i \leq n$, and an error message otherwise.
Note that there are no constraints on the order as long as the routine
consistently uses the same order in all calls. Using the access
routine with parameter $i$ is called \emph{accessing} $t_i$; the
time it takes to access a tuple is called \emph{access time} and
denoted $t_a$. 

\begin{definition} 
Let $d$ be a function from $\NNpos$ to $\RRpos$.
We define
$\RAgeneral{\lin}{d}$ to be the class of enumeration problems for
  which there exists a random-access algorithm $\mathcal{A}$ such that for every input
  $I$ the preprocessing phase takes time $t_p\in O(|I|)$ and the
  access time is
  $t_a \in O(d(|I|))$. 
Furthermore, $\RApolylog$ is the union of $\RAgeneral{\lin}{\log^c}$ for
all $c\geq 1$.
\end{definition}

Successively calling the access routine for $i=1,2,3,\ldots$ leads to:
\begin{fact}
\label{prop:RAisEN}
By definition, $\RAgeneral{\lin}{d}\subseteq\ENgeneral{\lin}{d}$ for all $d$.
\end{fact}

In the next subsection, we discuss the connection between the classes $\RAgeneral{\lin}{d}$ and
$\RPgeneral{\lin}{d}$.

\subsection{Random-Access and Random-Permutation}
We now show that, under certain conditions, it suffices to devise a
random-access algorithm in order to obtain a random-permutation algorithm. 
To achieve this, we need to produce a random permutation of the
indices of the answers.

Note that the trivial approach of producing the permutation upfront
will not work: the length of the permutation is the number of answers,
which can be much larger than the size of the input; however, we want
to produce the first answer after linear time in the size of the
input. 

Instead, we adapt a known random-permutation algorithm, the
\emph{Fisher-Yates Shuffle}~\cite{fisheryates:shuffle}, so that it works with
constant delay after constant preprocessing time.
The original version of the Fisher-Yates Shuffle (also known as \emph{Knuth Shuffle})
\cite{fisheryates:shuffle} generates a random permutation in time linear
in the number of items in the permutation, which in our setting is polynomial in the
size of the input. It initializes an array 
containing the numbers $0,\ldots,n{-}1$. Then, at each step $i$, it chooses a
    random index, $j$, greater than or equal to $i$ and swaps the chosen
    cell with the $i$th cell. At the end of this procedure, the array
    contains a random permutation. Proposition~\ref{prop:random_perm}
    describes an adaptation of this procedure that runs with constant
    delay and constant preprocessing time in the RAM model.

\begin{proposition} \label{prop:random_perm}
A random permutation of $0,\ldots,n{-}1$ can be generated with constant delay and constant preprocessing time.
\end{proposition}

\begin{algorithm}[t]
  \caption{Random Permutation}\label{alg:shuffle}
  \begin{algorithmic}[1]
  	\Procedure{Shuffle}{$n$}
  	\State assume $a[0],...,a[n{-}1]$ are uninitialized
  	\For {$i$ in $0,\ldots,n{-}1$}
        \State choose $j$ uniformly from $i,\ldots,n{-}1$
        \If {$a[i]$ is uninitialized}
    		\State $a[i] = i$
    	\EndIf
        \If {$a[j]$ is uninitialized}
    		\State $a[j] = j$
    	\EndIf
    	\State swap $a[i]$ and $a[j]$ ; output $a[i]$ 
    \EndFor
    \EndProcedure	
  \end{algorithmic}
\end{algorithm}

\begin{proof}
Algorithm~\ref{alg:shuffle} generates a random permutation with the required time constraints by simulating the Fisher-Yates Shuffle. Conceptually, it uses an array $a$ where at first all values are marked as ``uninitialized'', and an uninitialized cell $a[k]$ is considered to contain the value $k$.\footnote{To keep track of the array positions that are still uninitialized, one
can use standard methods for \emph{lazy array initialization} \cite{MoretShapiro}.}
At every iteration, the algorithm prints the next value in the permutation.

Denote by $a_j$ the value $a[j]$ if it is initialized, or $j$ otherwise.
We claim that in the beginning of the $i$th iteration, the values $a_i,\cdots,a_{n-1}$ are exactly those that the procedure did not print yet. This can be shown by induction: at the beginning of the first iteration, $a_0,\ldots,a_{n-1}$ represent $0,\ldots,n{-}1$, and no numbers were printed; at iteration $i{-}1$, the procedure stores in $a[i{-}1]$ the value that it prints, and moves the value that was there to a higher index.

At iteration $i$, the algorithm chooses to print uniformly at random a value between $a_i,\cdots,a_{n-1}$, so the printed answer at every iteration has equal probability among all the values that have not yet been printed. 
Therefore, Algorithm~\ref{alg:shuffle} correctly generates a random permutation.

The array $a$ can be simulated using a lookup table that is empty at
first and is assigned with the required values when the array
changes. In the RAM model with uniform cost measure, accessing such a
table takes constant time. Overall, Algorithm~\ref{alg:shuffle} runs
with constant delay, constant preprocessing time. Note that $O(n)$
space is used to generate a permutation of $n$ numbers.
\end{proof} 

With the ability to efficiently generate a random permutation of
$\{0,\ldots,n{-}1\}$, we can now argue that whenever we have
available a
random-access algorithm for an enumeration
problem and
if we can also tell the number of
answers, then we can build a random-permutation algorithm as follows:
we can produce, on the fly, a random permutation of the
indices of the answers and output each answer by using the access routine.

We say that an enumeration problem has \e{polynomially many answers} if the number
of answers per input $I$ is bounded by a polynomial in the size of
$I$. In particular, if $P$ is the evaluation of a CQ or a UCQ, then $P$
has polynomially many answers.

\begin{theorem}\label{thm:poly_RA_is_RP}
If $P\in\RAgeneral{\lin}{\log^c}$ and $P$ has polynomially many answers,
then $P\in\RPgeneral{\lin}{\log^c}$, for all $c\geq 1$.
\end{theorem}

\begin{proof}
Let $P$ be an enumeration problem in
$\RAgeneral{\lin}{\log^c}$, and 
let $\Algo{A}$ be the associated random-access algorithm for $P$.
When given an input $I$, our random-permutation algorithm proceeds as
follows.
It performs the preprocessing phase of $\Algo{A}$ and then, still
during its preprocessing phase, computes the number of answers $|P(I)|$ as follows.
We can tell whether $|P(I)|<k$ for any fixed $k$ by trying to access the
$k$th answer and checking if we get an out of bound error. We can use
this to do a binary search for the number of answers using
$O(\log(|P(I)|))$ calls to $\Algo{A}$'s access procedure.
Since $|P(I)|$ is polynomial in the size of the input,
$\log(|P(I)|)=O(\log(|I|))$. Each access costs time
$O(\log^c(|I|))$. In total, the number $|P(I)|$ is thus
computed in time $O(\log^{c+1}(|I|))$, which still is in $O(|I|)$.

During the enumeration phase, we use \Cref{prop:random_perm} to generate a random permutation of
$0,\ldots,|P(I)|{-}1$ with constant delay. 
Whenever we get the next element $i$ of the random permutation, we use
the access routine of $\Algo{A}$ to access the $(i{+}1)$th answer to our problem.
This procedure results in a random permutation of all the answers with
linear preprocessing time and delay $O(\log^c)$. 
\end{proof}

\section{Random-Access for CQs}\label{sec:cqs}
In this section, we discuss random access for CQs. For enumeration,
the characterization of CQs with respect to $\EN$ follows from known
results of Bagan, Durand, Grandjean, and Brault-Baron.

\begin{theorem}[\cite{bdg:dichotomy,bb:thesis}]\label{thm:CQ_EN}
Let $Q$ be a CQ. If $Q$ is free-connex, then it is in
$\ENgeneral{\lin}{\dconst}$. Otherwise, if it is also self-join-free, then it is not in
$\ENpolylog$ assuming $\sparseBMM$, $\Triangle$, and $\Hyperclique$.
\end{theorem}

Indeed, if the query $Q$ is self-join-free and
  non-free-connex, there are two cases. If $Q$ is cyclic, then it is
  not possible to determine whether there exists a first answer to $Q$
  in linear time assuming $\Triangle$ and $\Hyperclique$~\cite{bb:thesis}. Therefore,
  $Q$ it is not in $\ENgeneral{\lin}{\lin}$. Otherwise, if $Q$ is acyclic, the proof
  follows along the same lines as the one presented by Bagan et
  al.~\cite{bdg:dichotomy}.
  Using the same reduction as defined there, if any acyclic
  non-free-connex CQ is in $\ENgeneral{\lin}{\log^c}$, then any two Boolean matrices of
  size $n\times n$ can be multiplied in $O(m_1+m_2+m_3\cdot \log^c(n))$
  time, where $m_1$, $m_2$, and $m_3$ are the number of non-zero
  entries in $A$, $B$, and $AB$, respectively. This
  contradicts $\sparseBMM$.

According to \Cref{thm:CQ_EN}, free-connex CQs can be answered with
logarithmic delay. Brault-Baron~\cite{bb:thesis} proved that there
exists a random-access algorithm that works with linear preprocessing
and logarithmic access time. Hence, we get a strengthening of
\Cref{thm:CQ_EN}: free-connex CQs belong to $\RA$. According to
\Cref{thm:poly_RA_is_RP}, this also shows the tractability of a
random-order enumeration, that is, membership in $\RP$.

In this section, we present a random-access algorithm for free-connex
CQs that, compared to Brault-Baron~\cite{bb:thesis}, is simpler and better lends itself to a practical implementation. In
addition, we devise the algorithm in such a way that it is accompanied by an \e{inverted-access}
that is needed for our results on UCQs in Section~\ref{sec:ucqs}. An
inverted-access $I_A$ is an enhancement of a random-access algorithm
$A$ with the inverse operation: given an element $e$, the inverted-access returns $I_A[e]=j$ such that $A[j]=e$, that is, the $j$th
answer in the random-access is $e$; if $e$ is not an answer, then the
inverted-access indicates so by returning ``not-a-member.''

To proceed, we use the following folklore result.

\begin{proposition}\label{prop:reduction} For any free-connex CQ $Q$
  over a database $D$, 
  one can
  compute in linear time a full
  acyclic join query $Q'$ and a database $D'$ such that  $Q(D)=Q'(D')$
  and $D'$ is
  globally
  consistent w.r.t.\ $Q'$.
\end{proposition}
This reduction was implicitly used in the past as part of CQ answering
algorithms~(cf., e.g., \cite{cdy,DBLP:journals/tods/OlteanuZ15}). 
To prove it, the first step is performing a full
reduction to remove dangling tuples (tuples that do not agree with any
answer) from the database. This can be done in linear time as proposed
by Yannakakis~\cite{Yannakakis} for acyclic join queries. 
Then, we utilize the fact that $Q$ is \emph{free-connex}, which enables us to 
drop all atoms that contain quantified variables.
This leaves us with a full acyclic join that has the same answers as the
original free-connex CQ.

So, it is left to design a random-access algorithm for full acyclic CQs. We do
so in the remainder of this section.
Algorithm~\ref{alg:preprocess}
  describes the preprocessing phase that builds the data structure
  and computes the count (i.e., the number $|Q(D)|$ of answers). Then, Algorithm~\ref{alg:access} 
  provides random-access to the answers,  
  and Algorithm~\ref{alg:index} provides
  inverted-access.

\begin{algorithm}[t]
  \caption{Preprocessing}\label{alg:preprocess}
  \begin{algorithmic}[1]
        \Procedure{Preprocessing}{$D, Q$}
        \For{$R$ in leaf-to-root order}
            \State Partition $R$ to buckets according to $\patts{R}$
            \For{bucket $B$ in $R$}
                \For{tuple $t$ in $B$}
                                \If{$R$ is a leaf}
                                        \State $\wgt(t) = 1$
                                        \Else
                                        \State let $C$ be the children of $R$
                                        \State $\wgt(t) = \prod_{S\in C}{\wgt(\bucket{S}{t})}$
                                \EndIf
                                \State let $P$ be the tuples preceding $t$ in $B$
                                \State $\rng(t)= \sum_{s\in P}{\wgt(s)}$
                \EndFor
                \State $\wgt(B) = \sum_{t\in B}{\wgt(t)}$
            \EndFor
        \EndFor
      \EndProcedure     
  \end{algorithmic}
\end{algorithm}

Given a relation $R$, denote by $\patts{R}$ the attributes that appear both in $R$ and in its parent. If $R$ is the root, then $\patts{R}=\emptyset$.
Given a relation $R$ and an assignment $a$, we denote by $\bucket{S}{a}$ all tuples in $S$ that agree with $a$ over the attributes that $S$ and $a$ have in common. We use this notation also when $a$ is a tuple, by treating the tuple as an assignment from the attributes of its relation to the values it holds (intuitively this is $S\ltimes a$).

The preprocessing starts by partitioning every relation to buckets
according to the different assignments to the attributes shared with
the parent relation. This can be done in linear time in the RAM model.
Then, we compute a weight $\wgt(t)$ for each tuple $t$. This weight
represents the number of different answers this tuple agrees with when
only joining the relations of the subtree rooted in the current
relation. The weight is computed in a leaf-to-root order, where tuples
of a leaf relation have weight $1$. The weight of a tuple $t$ in a
non-leaf relation $R$ is determined by the product of the weights of
the corresponding tuples in the children's relations. These
corresponding tuples are the ones that agree with $t$ on the
attributes that $R$ shares with its child. The weight of each bucket
is the sum of the weights of the tuples it contains. In addition, we
assign each tuple $t$ with an index range that starts with $\rng(t)$
and ends with the $\rng$ of the following tuple in the bucket (or the
total weight of the bucket if this is the last tuple). This represents
a partition of the indices from $0$ to the bucket weight, such that
the length of the range of each tuple is equal to its weight. At the
end of preprocessing, the root relation has one bucket (since
$\patts{\root}=\emptyset$), and the weight of this bucket represents
the number of answers to the query.

\begin{algorithm}[t]
  \caption{Random-Access}\label{alg:access}
  \begin{algorithmic}[1]
        \Procedure{Access}{$j$}
        \If {$j \geq \wgt(\root)$}
            \State return out-of-bound
        \Else
            \State $\ans = \emptyset$
            \State \algname{SubtreeAccess($\root$, $j$)}
            \State return $\ans$
        \EndIf
        \EndProcedure
      \item[]
        \Procedure{SubtreeAccess}{$R, j$}
        \State find $t\in R$ s.t.~$\rng(t) \leq j < \rng(t{+}1)$
        \State $\ans = \ans \cup \{\atts{R}\rightarrow \atts{R}(t)\}$
        \State let $R_1,\ldots,R_m$ be the children of $R$
        \State $j_1,\ldots, j_m = \split(j-\rng(t),$\\\hskip8em $\wgt(\bucket{R_1}{t}),\ldots,\wgt(\bucket{R_m}{t}))$ \label{line:split}
        \For{$i$ in $1,\ldots,m$}
            \State \algname{SubtreeAccess($\bucket{R_i}{t}$, $j_i$)}
        \EndFor
      \EndProcedure
  \end{algorithmic}
\end{algorithm}

The random-access is done recursively in a root-to-leaf order: we
start from the single bucket at the root. At each step we find the
tuple $t$ in the current relation that holds the required index in its
range (we denote by $t$+1 the tuple that follows $t$ in the
bucket). Then, we assign the rest of the search to the children of the
current relation, restricted to the bucket that corresponds to $t$.
Finding $t$ can be done in logarithmic time using binary search. The
remaining index $j'=j-\rng(t)$ is split into search tasks for the
children using the method $\split$. This can be done
in the same way
as
an index is split in standard multidimensional arrays: if the
last bucket is of weight $m$, its index would be $j' \text{ mod } m$,
and the other buckets will now recursively split between them the
index $\lfloor\frac{j'}{m}\rfloor$.

Algorithm~\ref{alg:index} works similarly to
Algorithm~\ref{alg:access}. But while the search down the tree in
Algorithm~\ref{alg:access} is guided by the index and the answer is
the assignment, in Algorithm~\ref{alg:index} the search
is guided by the assignment and the 
answer is the index. 
The function $\combine$ is the reverse of $\split$, used in line~\ref{line:split} of Algorithm~\ref{alg:access}.
Recursively,
$\combine(w_1,j_1,\ldots,w_m,j_m)$ is given by $j_m+w_m\cdot
\combine(w_1,j_1,\ldots,w_{m-1},j_{m-1})$ with $\combine()\!=\!0$.

\begin{algorithm}[t]
  \caption{Inverted-Access}\label{alg:index}
  \begin{algorithmic}[1]
        \Procedure{InvertedAccess}{$\ans$}
        \State return \algname{InvertedSubtreeAccess($\root$, $\ans$)}
        \EndProcedure
      \item[]
        \Procedure{InvertedSubtreeAccess}{$R$, $\ans$}
        \State find $t\in R$ s.t.~$\atts{R}(t) = \atts{R}(\ans)$ \label{line:find_t}
        \If{$t$ was not found}
            \State return not-an-answer
        \EndIf
        \State let $R_1,\ldots,R_m$ be the children of $R$
        \For{$i$ in $1,\ldots,m$}
            \State $j_i$ = \algname{InvertedSubtreeAccess($R_i$, $\ans$)}
        \EndFor
        \State $\offset = \combine(
        \wgt(\bucket{R_1}{\ans}),j_1, \ldots, $ \\\hskip11.4em
        $\wgt(\bucket{R_m}{\ans}),j_m)$
        \State return $\rng(t) + \offset$
        \EndProcedure
  \end{algorithmic}
\end{algorithm}

Line~\ref{line:find_t} can be supported in constant time after an
appropriate indexing of the buckets at preprocessing. Since
Algorithm~\ref{alg:index} has a constant number of operations (in data
complexity), inverted-access can be done in constant time (after the
linear preprocessing provided by Algorithm~\ref{alg:preprocess}).

The next theorem, 
parts of which are already given in \cite{bb:thesis},
summarizes the algorithms presented so far.

\begin{theorem}\label{thm:cqops} Given a free-connex CQ 
$Q$ and a database $D$, it is
possible to build in linear time a data structure that allows
to output the count $|Q(D)|$ in constant time and provides random-access in logarithmic time, and
inverted-access in constant time. 
\end{theorem}

\begin{example}
  Consider the CQ
  \[Q(v,w,x,y,z)\cqa R_1(v,w,x), R_2(v,y), R_3(w,z)\] with the
  join-tree with $R_1$ as root, and $R_2$ and $R_3$ are its children. The
  following is an example of an input database for such a query and the
  computed information available at the end of preprocessing.
  Here, the $\rng$ value is denoted $s$. 
  \begin{table}[h] \begin{tabular}{|c c
        c|c|c|} \hline\small
                                                   $R_1$ & & & $\wgt$ & $s$\\
                                                   \hline\hline
                                                   $a_1$ & $b_1$ & $c_1$ & $6$ & $0$\\
                                                   $a_1$ & $b_1$ & $c_2$ & $2$ & $6$\\
                                                   $a_2$ & $b_2$ & $c_1$ & $6$ & $8$\\
                                                   $a_2$ & $b_2$ & $c_2$ & $2$ & $14$\\
                                                   \hline
                                                 \end{tabular}
                                                \begin{tabular}{|c
                                                  c|c| c|} \hline
                                                  $R_2$ & & $\wgt$ & $s$\\
                                                  \hline\hline
                                                  $b_1$ & $d_1$ & $1$ & $0$\\
                                                  $b_1$ & $d_2$ & $1$ & $1$\\
                                                  \hline
                                                  $b_2$ & $d_2$ & $1$ & $0$\\
                                                  $b_2$ & $d_3$ & $1$ & $1$\\
                                                  \hline \end{tabular}
                                                \begin{tabular}{|c
                                                  c|c| c|} \hline
                                                  $R_3$ & & $\wgt$ & $s$\\
                                                  \hline\hline
                                                  $c_1$ & $e_1$ & $1$ & $0$\\
                                                  $c_1$ & $e_2$ & $1$ & $1$\\
                                                  $c_1$ & $e_3$ & $1$ & $2$\\
                                                  \hline
                                                  $c_2$ & $e_4$ & $1$ & $0$\\
                                                  \hline \end{tabular}
                                                \end{table}

                                                Calling
                                                $\algname{Access}(13)$
                                                finds
                                                $(a_2,b_2,c_1)\in R_1$.
                                                Then, the remaining
                                                $13-8=5$ is split to
                                                $5\text{ mod } 3 = 2$ in the
                                                top bucket of $R_3$
                                                and
                                                $\lfloor\frac{5}{3}\rfloor
                                                = 1$ in the bottom
                                                bucket of $R_2$. These
                                                in turn find the
                                                tuples
                                                $(b_2,d_3)\in R_2$ and
                                                $(c_1,e_3)\in R_3$.
                                                Overall, the obtained
                                                answer is
                                                $(a_2,b_2,c_1,d_3,e_3)$.

                                                Calling
                                                $\algname{InvertedAccess}(a_2,b_2,c_1,d_3,e_3)$
                                                finds
                                                $(a_2,b_2,c_1)\in R_1$
                                                with $\rng=8$. Then calling
                                                \algname{InvertedSubtree\-Access}
                                                on $R_2$ returns the index
                                                $\rng(b_2,d_3)=1$ from a bucket of weight $2$, and calling
                                                \algname{InvertedSubtree\-Access}
                                                on $R_3$ returns
                                                $\rng(c_1,e_3)=2$ from a bucket of weight $3$. The
                                                call for
                                               \algname{Combine\-Index}$(2,1,3,2)$
                                                returns 
                                                $2+3\cdot 1 = 5$, 
                                                and the result is
                                                $8+5=13$. \qed
                                              \end{example}

\Cref{thm:cqops} along with \Cref{thm:poly_RA_is_RP} implies that the dichotomy of \Cref{thm:CQ_EN} extends to the problems $\RP$ and 
\linebreak[4]
$\RA$. This also means that for
self-join-free
CQs, the classes of efficient enumeration, random-access and random-permutation collapse.
This is summarized by the next corollary.

\def\corCqEquiv{
For every CQ $Q$, the following holds:
 If $Q$ is free-connex, then $Q$ is in each of $\RA$, $\RP$ and $\EN$. 
 If $Q$ is self-join-free and not free-connex, then it is not in any of $\RA$, $\RP$, and $\EN$ assuming $\sparseBMM$, $\Triangle$, and $\Hyperclique$.
}
\begin{corollary} \label{corollary:cqEquiv}
\corCqEquiv
\end{corollary}

\section{Unions of CQs}\label{sec:ucqs}

In this section, we discuss the availability of random-order enumeration and
random-access in UCQs. We first show that not all UCQs that have
efficient enumeration also have efficient random-access. Then we relax
the delay requirements and provide an algorithm for the enumeration in
random order of a union of sets, and show that the algorithm can be
applied for such UCQs.
In addition, we identify a subclass of UCQs that do allow for an efficient random-access.

If several CQs are in $\ENgeneral{\lin}{d}$, for some $d$, then their union can also be enumerated
within the same time bounds~\cite{DBLP:conf/csl/DurandS11,nofar:ucq}. Since our goal is to
answer queries in random order, a natural question arises: does the
same apply to queries in $\RAgeneral{\lin}{d}$ and $\RPgeneral{\lin}{d}$? We show that it does not
apply to CQs in $\RAgeneral{\lin}{d}$. This means that for UCQs we cannot rely on
random-access to achieve an efficient random-permutation algorithm as
we did for CQs.
The following is an example of two free-connex CQs (therefore, each one
admits efficient counting, enumeration, random-order enumeration and
random-access), but we show that their union is not in $\RAgeneral{\lin}{\lin}$
under $\Triangle$.

\begin{example}\label{ex:cyclic_intersection}
Consider the CQs
$Q_1(x,y,z) \cqa R(x,y),S(y,z)$ \ and \ 
$Q_2(x,y,z) \cqa S(y,z),T(x,z)$.
Let $Q_\cup = Q_1\cup Q_2$. Since $Q_1$ and $Q_2$ are both
free-connex, we can find $|Q_1(D)|$ and $|Q_2(D)|$ in linear time 
by Theorem~\ref{thm:cqops}.
Note that
$|Q_\cup(D)|=|Q_1(D)|+|Q_2(D)|-|Q_1(D)\cap Q_2(D)|$. Therefore,
$|Q_1(D)\cap Q_2(D)| > 0$ iff $|Q_\cup(D)| < |Q_1(D)| + |Q_2(D)|$.

Now let us assume that $Q_\cup\in\RAgeneral{\lin}{\lin}$. 
We can then ask the random-access algorithm for $Q_\cup$ to retrieve index number
$|Q_1(D)| + |Q_2(D)|$. The algorithm will raise an out-of-bound error
exactly if $|Q_\cup(D)| < |Q_1(D)| + |Q_2(D)|$. Therefore, we can
check whether $Q_1(D)\cap Q_2(D) = \emptyset$ in linear time.
But consider the ``triangle query'' $Q_\cap(x,y,z) \cqa R(x,y), S(y,z),
T(x,z)$ and note that
$Q_\cap(D)=Q_1(D)\cap Q_2(D)$ for all $D$.  We can hence
determine if the query $Q_\cap$ has answers in linear time, which
contradicts $\Triangle$.
Thus, under $\Triangle$, the UCQ $Q_\cup$ does not belong to 
$\RAgeneral{\lin}{\lin}$.
\end{example}

Example~\ref{ex:cyclic_intersection} shows that (assuming $\Triangle$)
$\RA$ is not closed under union. It also shows
that, when considering UCQs, we have that
$\ENgeneral{\lin}{\dconst}\not\subseteq\RAgeneral{\lin}{\lin}$. In particular, this means that $\EN\neq\RA$, which is not the case when only considering CQs.
In
Section~\ref{subsec:UCQsRandomAccess}, we devise a sufficient
condition for UCQs to have a $\RAgeneral{\lin}{\polylog}$ computation.
In
Section~\ref{sec:expected}, we show that if we relax the bound to logarithmic
time \emph{in expectation}, we can enumerate in a random-order any union comprised of free-connex CQs.

\subsection{Random-Permutation with Expected Logarithmic Delay}\label{sec:expected}

In order to provide a random-permutation algorithm for UCQs, we start
by devising a general algorithm for the union of sets, and then show
how it can be applied for UCQs. The sets are assumed to have efficient
counting, uniform sampling, membership testing, and deletion. If the
number of sets in the union is constant, the algorithm also carries
the guarantees of expected and amortized constant number of such
operations between every pair of successively printed answers.
The algorithm
resembles
the sampling algorithm by Karp and
Luby~\cite{karp-luby}, but it allows for sampling without repetitions.
We prove the following lemma.

\begin{lemma}\label{lemma:union}
  Let $S_1,\ldots,S_k$ be sets, each supports sampling, testing,
  deletion and counting in time $t$. Then, it is possible to enumerate
  $\bigcup_{j=1}^{k}{S_j}$ in uniformly random order with expected
  $O(kt)$ delay.
\end{lemma}

\begin{algorithm}[t]
  \caption{Random-Order Enumeration of 
$S_1\cup\cdots\cup S_k$
}\label{alg:union}
  \begin{algorithmic}[1]
        \While {$\sum_{j=1}^{k}{S_j.\algname{Count}()} > 0$} \label{line:while}
        \State $\chosen$ = choose $i$ with probability $\frac{S_i.\algname{Count}()}{\sum_{j=1}^{k}{S_j.\algname{Count}()}}$ \label{line:choose_set}
        \State $\element = S_\chosen.\algname{Sample}()$ \label{line:choose_element}
        \State $\providers = \{ S_j \mid S_j.\algname{Test}(\element) = \true \}$ \label{line:providers}
        \State $\owner = \min\{j \mid S_j\in\providers\}$ \label{line:owner}
        \For{$S_j \in \providers\setminus \{S_{\owner}\}$} \label{line:providers_loop}
            \State $S_j.\algname{delete}(\element)$ \label{line:providers_deletion}
        \EndFor
        \If{$S_{\owner} = S_\chosen$} \label{line:owner_if}
                        \State
                        $S_\chosen.\algname{delete}(\element)$ \label{line:owner_deletion}
                        ; output $\element$ \label{line:print}
        \EndIf
        \EndWhile
  \end{algorithmic}
\end{algorithm}

\Cref{alg:union} enumerates the union of several sets in uniformly random order. Every iteration starts by choosing a random set and a random element it contains. The choice of set is weighted by the number of elements it contains.
If the algorithm would have always printed the element at that stage (after line~\ref{line:choose_element}), then an element that appears in two sets would have had twice the probability of being chosen compared to an element that appears in only one set. The following lines correct this bias.
We denote by $\providers$ all sets that contain the chosen element.  Then, the algorithm assigns one owner to this element out of its providers (as the choice of the owner is not important, we arbitrarily choose to take the provider with the minimum index).
The element is then deleted from non-owners, and is printed only if the algorithm chooses its owner in line~\ref{line:choose_set}.
If the element was reached through a non-owner, then the current iteration ``rejects'' by printing nothing.

\def \Choices{\mathit{Choices}}
\def \AccChoices{\mathit{AccChoices}}

Algorithm~\ref{alg:union} prints the results in a uniformly random
order since, in every iteration, every answer remaining in the union
has equal probability of being printed. 
Denote by $\Choices$ the set of all possible $(\chosen, \element)$
pairs that the algorithm may choose in lines~\ref{line:choose_set}
and~\ref{line:choose_element}. The probability of such a pair is
$\frac{|S_\chosen|}{\sum_{j=1}^{k}{|S_j|}}\frac{1}{|S_\chosen|} =
\frac{1}{\sum_{j=1}^{k}{|S_j|}}$, which is the same for all pairs in
$\Choices$. 
Denote by $\AccChoices\subseteq \Choices$ the pairs for which
$S_\chosen$ is the owner of $\element$. Line~\ref{line:owner_if}
guarantees that an element is printed only when the selections the
algorithm makes are in $\AccChoices$. Since every possible answer only
appears once as an element in $\AccChoices$, the probability of each
element to be printed is $\frac{1}{\sum_{j=1}^{k}{|S_j|}}$. Therefore,
all answers have the same probability of being printed. A printed
answer is deleted from all sets containing it, so it will not be
printed twice.

We now discuss the time complexity.
If some iteration rejects an answer, this iteration also deletes it from all non-owner sets. This guarantees that each unique answer will only be rejected once, as it only has one provider in the second time it is seen. This means that the total number of iterations Algorithm~\ref{alg:union} performs is bound by twice the number of answers. The delay between successive answers is therefore amortized constant.
In addition, since by definition $|\Choices|\leq k|\AccChoices|$, in every iteration the probability that an answer will be printed is $\frac{|\AccChoices|}{|\Choices|}\geq \frac{1}{k}$.
The delay between two successive answers therefore comprises of a constant number of operations both in expectation and in amortized complexity.
This proves \Cref{lemma:union}.

In order to use Algorithm~\ref{alg:union}, the sets in question need to support counting, sampling, testing and deletion. We next show how to support these operations using the shuffle mechanism provided in Algorithm~\ref{alg:shuffle}, assuming that the sets in question support efficient counting, random-access and inverted-access.
Then, we will be able use this algorithm to answer UCQs.

We describe the construction of the data structure.
First, we count the number of answers $n$. As in Algorithm~\ref{alg:shuffle}, our data structure contains an array $a$ of length $n$ and an integer $i$. Here, $i$ corresponds to the number of elements deleted. The values $a[0],\ldots,a[i{-}1]$ represent the indices of the deleted elements, while $a[i],\ldots,a[n{-}1]$ hold the indices that remain in the set. We also use a reverse index $b$: whenever we set $a[i]=j$, we also set $b[j]=i$. 
Conceptually, $a$ and $b$ start initialized with $a[j]=b[j]=j$ and $i=0$. Practically, the arrays can be implemented as lookup tables as in Algorithm~\ref{alg:shuffle}.
When \e{sampling}, we generate a uniformly random number $k\in\{i,\ldots,n{-}1\}$. We then return element number $a[k]$ using the random-access routine.
When \e{testing} membership, we call the inverted-access routine and return ``True'' iff we obtain a valid index.
When \e{deleting}, we use the inverted-access routine to find the index $m$ of the item to be deleted. We then find an index $k$ such that $a[k]=m$, swap $a[k]$ with $a[i]$, and increase $i$ by one. In order to efficiently find $k$, we use the reverse index $b$.
When \e{counting}, we return $n{-}i$.
The correctness of these procedures follows along the same lines of that of Algorithm~\ref{alg:shuffle}.
This proves the following lemma.

\begin{lemma}\label{lemma:deletion}
If an enumeration problem supports counting, random-access and inverted-access in time $t$, then the set of its answers also supports sampling, testing, deletion and counting in time $O(t)$.
\end{lemma}
 
Since free-connex CQs admit efficient algorithms for counting, random-access and inverted-access, we can apply this result to UCQs.
Combining \Cref{thm:cqops} with \Cref{lemma:union} and
\Cref{lemma:deletion}, we have an algorithm for answering UCQs with
random order.

\begin{theorem}
Let $Q$ be a union of free-connex CQs. There exists a random-permutation algorithm for answering $Q$ that uses linear preprocessing and expected logarithmic delay.
\end{theorem}

\subsection{UCQs that Allow for Random-Access}\label{subsec:UCQsRandomAccess}

We now identify a class of UCQs that allow for
random-access with polylogarithmic access time and linear
preprocessing
(and hence, via Theorem~\ref{thm:poly_RA_is_RP} also allow for random-order enumeration with linear preprocessing and polylogarithmic delay).

Assume two sets $A$ and $A'$ such that $A'\subseteq A$.  An order over
$A'$ is \e{compatible} with an order over $A$ if the former is a
subsequence of the latter, that is, the precedence relationship of the
elements of $A'$ is the same in both orders.  A \emph{mutually
  compatible} UCQ, or \emph{mc-UCQ} for short, is a UCQ
$Q=Q_1\cup\cdots\cup Q_m$
such that for all $\emptyset\neq I\subseteq [1,m]$, the CQ
$Q_{I}\deff \bigcap_{i\in I}Q_i$ is free-connex and, moreover, there
are $\RA$-algorithms $\Algo{A}_I$ for $Q_I$ that:
\e{(a)} 
provide inverted access in logarithmic time;
\e{(b)}
are \emph{compatible} in the sense that on every database $D$
  and $\emptyset\neq I\subseteq [1,m]$ we have that $\Algo{A}_{I}$ is
  compatible with $\Algo{A}_{\set{\min(I)}}$.
 We can prove the following.

  \def\thmRAforUCQs{
Every mc-UCQ $Q$ belongs to $\RAlogsquare$ and to $\RPlogsquare$.
    }
    \begin{theorem}\label{thm:RAforUCQs}
      \thmRAforUCQs
\end{theorem}

An example of an mc-UCQ is 
$Q_7^S\cup Q_7^C$ used in the experiments in \Cref{sec:experiments}. This UCQ is comprised of two acyclic CQs with the same structure, except they use different relations (formed by different selections applied on the same initial relations). These CQs have the following structure for $i\in\{S,C\}$:
$Q_7^i(o,c,a,b,p,s,l,m,n)\cqa$
$R(s,a),
L(o,p,s,l),
O(o,c),
B(c,b),
N^i(a,m),
M^i(b,n)$.
Applying \Cref{thm:cqops} on $Q_7^S$, $Q_7^C$ and
$Q_7^S\cap Q_7^C$, we can construct algorithms for random-access and inverted-access in a compatible order.

The remainder of this section describes the algorithm for
proving Theorem~\ref{thm:RAforUCQs}. By Theorem~\ref{thm:poly_RA_is_RP} we can focus on $\RAlogsquare$.

\subsubsection*{Random-access for unions of sets}
We start with the abstract setting of providing random-access for a
union of sets (of arbitrary elements) and then turn to the specific
setting where these sets are the results of the CQs that a given UCQ
consists of.

We build upon Durand and Strozecki's \e{union
  trick}~\cite{DBLP:conf/csl/DurandS11}, which can be described as
follows.  Assume that $A$ and $B$ are two (not necessarily disjoint)
subsets of a certain universe $U$, and for each of these sets, we have
available an algorithm that enumerates the elements of the set.
Furthermore, assume that for the set $B$ we also have available an
algorithm for testing membership in $B$. The goal is to enumerate
$A\cup B$ (and, as usual, all enumerations are without repetitions).
The pseudocode for the union trick is provided in
Algorithm~\ref{alg:DS-UnionTrick}.  Here, ``$a$ =
$A$.\algname{First()}'' means that the enumeration algorithm for $A$
is started and $a$ shall be the first output element. Similarly, ``$a$
= $A$.\algname{Next}()'' means that the next output element of the
enumeration algorithm for $A$ is produced and that $a$ shall be that
element. In case that all elements of $A$ have already been
enumerated, $A.\algname{Next}()$ will return the end-of-enumeration
message $\EOE$; and in case that $A$ is the empty set,
$A.\algname{First}()$ will return $\EOE$.

\def\output{output\xspace}

\begin{algorithm}[t]
\caption{Durand-Strozecki's Union Trick for $A\cup B$}\label{alg:DS-UnionTrick}
 \begin{algorithmic}[1]
   \State $a$ = $A.\algname{First}()$ ;  $b$ = $B.\algname{First}()$
  \While{$a \neq \EOE$}
    \If{$a \not\in B$} 
      \State{\output $a$ ; $a$ = $A.\algname{Next}()$}
    \Else
       \State{\output $b$ ; $b$ = $B.\algname{Next}()$ ; $a$ = $A.\algname{Next}()$}
    \EndIf
  \EndWhile
  \While{$b\neq\EOE$}
    \text{$\{$ \output $b$ ; $b$ = $B.\algname{Next}()$ $\}$}
  \EndWhile
\end{algorithmic}
\end{algorithm}

\begin{algorithm}[htbp]
\caption{Random-access for $A\cup B$}\label{alg:AccessUnionOfTwo}
 \begin{algorithmic}[1]
  \Function{($A\cup B$).\Access}{$j$}
   \State $a$ = $A$.\Access($j$)
   \If{$a \neq \ErrorMessage$}
     \If{$a\not\in B$}
       \State \output $a$
     \Else
        \State 
         $k$ = $(A\cap B)$.\Rank($a$) ; \label{line:computek}
         $b$ = $B$.\Access($k$) ;
         \output $b$
     \EndIf
   \Else
     \ \text{
       $\ell$ = $j - |A| + |A\cap B|$ ; 
       $b$ = $B$.\Access($\ell$) ;
       output $b$}
   \EndIf  
  \EndFunction
 \end{algorithmic}
\end{algorithm}

This algorithm starts by enumerating all elements of $A$ in the same
order as the enumeration algorithm for $A$, but every time it
encounters $a\in A\cap B$, it ignores this element and instead outputs
the next available element produced by the enumeration algorithm for
$B$. Once the enumeration of $A$ has terminated, the algorithm
proceeds by producing the remaining elements of $B$.  Clearly,
Algorithm~\ref{alg:DS-UnionTrick} enumerates all elements in
$A\cup B$; and the algorithm's delay is $O(d_A+d_B+t_B)$ where $d_A$
and $d_B$ are the delay of the enumeration algorithms for $A$ and $B$,
respectively, and $t_B$ is the time needed for testing membership in
$B$.

The idea is to provide random-access to the $j$th output element
produced by Algorithm~\ref{alg:DS-UnionTrick}. 
Let us write $a_1,a_2,\ldots,a_n$ and $b_1,b_2,\ldots,b_{n'}$ for the
elements of $A$ and $B$, repectively, as they are produced by the
given enumeration algorithms for $A$ and for $B$.  Let us first
consider the case where $j\leq |A|$.  Clearly, the $j$th output
element of Algorithm~\ref{alg:DS-UnionTrick} will be $a_j$ if
$a_j\not\in B$; and in case that $a_j\in B$, the $j$th output element
of Algorithm~\ref{alg:DS-UnionTrick} will be $b_k$ for the particular
number $k=|\set{a_1,\ldots,a_j}\cap B|$.  In case that $j>|A|$, the
$j$th output element of Algorithm~\ref{alg:DS-UnionTrick} will be
$b_\ell$ for $\ell= j-|A|+|A\cap B|$.

But how can we compute $k=|\set{a_1,\ldots,a_j}\cap B|$
efficiently upon input of $j$?  Following is a
sufficient condition. 
Assume we have available an algorithm that enumerates $A\cap B$, and
its enumeration order is \emph{compatible} with that of the
enumeration algorithm for $A$ in the sense defined above.
Furthermore, assume
that we have available a routine
``$(A\cap B)$.\Rank$(c)$'' that, upon input of an arbitrary
$c\in A\cap B$ returns the particular number $i$ such that $c$ is the
$i$th element produced by the enumeration algorithm for $A\cap B$. (We
say that $i$ is the \emph{rank} of $c$ in $A\cap B$.)
Then we can compute $k=|\set{a_1,\ldots,a_j}\cap B|$ by using that
$|\set{a_1,\ldots,a_j}\cap B|= (A\cap B)\text{.\Rank}(a_j)$.
This immediately leads to the random-access algorithm for $A\cup B$
whose pseudocode is given in Algorithm~\ref{alg:AccessUnionOfTwo}.

Our next goal is to generalize this to the union of $m$ sets
$S_1,\ldots,S_m$ for an arbitrary $m\geq 2$. We proceed by induction
on $m$ and have already established the basis for $m=2$. Let us now
consider the induction step from $m-1$ to $m$.  We let $A= S_1$ and
$B=S_2\cup\cdots\cup S_m$ and use Algorithm~\ref{alg:DS-UnionTrick} to
enumerate $A\cup B=S_1\cup\cdots\cup S_m$, where the routines
$B.\algname{First}()$ and $B.\algname{Next}()$ are provided by the
induction hypothesis.  We would like to use
Algorithm~\ref{alg:AccessUnionOfTwo} to provide random-access to the
$j$-th element that will be enumerated from $A\cup B$.  By assumption,
we know how to compute $|A|$ and $a=A.\Access(j)$; and by the
induction hypothesis, we already know how to compute
$b=B.\Access(j)$. What we still need in order to execute
Algorithm~\ref{alg:AccessUnionOfTwo} is a way to compute $|A\cap B|$
and a workaround with which we can replace the command
$k=(A\cap B).\Rank(a)$; recall that this command was introduced to
compute the number $k=|\set{a_1,\ldots,a_j}\cap B|$.

Computing $|A\cap B|$ for $A=S_1$ and $B=S_2\cup\cdots\cup S_m$ is
easy: we can use the inclusion-exclusion principle and obtain $|A\cap
B|=$
\[
\begin{array}{rcl}
 \displaystyle \left| \,\bigcup_{i=2}^m (S_1\cap S_i)\,\right|
& =
& \displaystyle \sum_{\emptyset\neq I\subseteq[2,m]} (-1)^{|I|+1}
\left|\, \bigcap_{i\in I}(S_1\cap S_i)\, \right| .
\end{array}
\]
Thus, we can compute $|A\cap B|$ provided that for each $I\subseteq [2,m]$, we can compute the cardinality $|T_{1,I}|$ of the set
$T_{1,I} 
 \deff 
 S_1 \cap \bigcap_{i\in I} S_i$.

Let us now discuss how to compute $k=|\set{a_1,\ldots,a_j}\cap B|$.
Again using the inclusion-exclusion principle, we obtain that
\[|\set{a_1,\ldots,a_j}\cap B| =
 \displaystyle \sum_{\emptyset\neq I\subseteq[2,m]} (-1)^{|I|+1}
 \left|\, \bigcap_{i\in I}(\set{a_1,\ldots,a_j}\cap S_i)\, \right| .
\]
We can compute this number if for each
$\emptyset\neq I\subseteq [2,m]$ we can compute 
$n_{1,I} \ \deff \ \left|\,\set{a_1,\ldots,a_j} \cap \bigcap_{i\in
    I}S_i \,\right|$.  To compute $n_{1,I}$, assume we have
available an algorithm that enumerates $T_{1,I}$, and its enumeration
order is compatible with that of the algorithm for
$A=S_1$.  Furthermore, assume we have available a routine
$T_{1,I}.\Rank(c)$ that, given $c\in T_{1,I}$, returns the particular
$i$ such that $c$ is the $i$th element produced by the enumeration
algorithm for $T_{1,I}$.  In addition, assume that we have available a
routine $T_{1,I}.\Largest(a)$ that, given $a\in S_1$, returns the
particular $c\in T_{1,I}$ such that $c$ is the largest element of
$T_{1,I}$ 
that is less than or equal to $a$ 
in the enumeration order of $S_1$.
Then, we can compute $n_{1,I}$ by using that
$n_{1,I}=T_{1,I}.\Rank(b)$ for $b\deff T_{1,I}.\Largest(a_j)$.  In
summary, we can replace the first command in line~7 of
Algorithm~\ref{alg:AccessUnionOfTwo} by Algorithm~\ref{alg:Computek}.

\begin{algorithm}[t]
\caption{Workaround for line~7 of Algorithm~\ref{alg:AccessUnionOfTwo} for $S_1\cup\cdots\cup S_m$. We assume that $a\in S_1$.}\label{alg:Computek}
 \begin{algorithmic}[1]
  \Procedure{Compute-$k$ }{$a$}
   \For{each $I$ with $\emptyset\neq I\subseteq[2,m]$}
      \State $b = T_{1,I}$.\Largest($a$) ; $n_{1,I} = T_{1,I}$.\Rank($b$)
   \EndFor
   \State $k= \sum_{\emptyset\neq I\subseteq [2,m]}(-1)^{|I|+1}n_{1,I}$   
   \ ; \ output $k$
  \EndProcedure
 \end{algorithmic}
\end{algorithm}

To recap, we obtain the following for
$S_1\cup\cdots\cup S_m$. For $\ell\in [1,m]$ and each $I$ with
$\emptyset\neq I\subseteq [\ell{+}1,m]$, let
$T_{\ell,I}\deff S_{\ell}\cap\bigcap_{i\in I}S_i$.  Assume that for
every $\ell\in[1,m]$ we have available an enumeration algorithm for
$S_\ell$, and for every $\emptyset\neq I\subseteq [\ell{+}1,m]$ we
have available an enumeration algorithm for $T_{\ell,I}$, so that all
of the following hold.
\begin{enumerate}
\item The enumeration for $T_{\ell,I}$ is compatible with that for
  $S_\ell$.
\item After having carried out the preprocessing phase for $S_\ell$:
  \e{(a)} we know its cardinality $|S_\ell|$, \e{(b)} given $j$, the
  routine $S_\ell.\Access(j)$ returns in time $t_{\textit{acc}}$ the
  $j$th output element of the enumeration algorithm for $S_\ell$,
  \e{and (c)} given $u$, it takes time $t_{\textit{test}}$ to test whether $u\in S_\ell$.
\item After having carried out the preprocessing phase for
  $T_{\ell,I}$: \e{(a)} we know its cardinality $|T_{\ell,I}|$,
  \e{(b)} given $c\in T_{\ell,I}$, the rank $T_{\ell,I}.\Rank(c)$ can
  be computed in time $t_{\textit{inv-acc}}$, \e{and (c)} given
  $a\in S_\ell$, it takes time $t_{\textit{lar}}$ to return largest
  element of $T_{\ell,I}$ that does not succeed $a$ in the
  enumeration order of $S_\ell$.
\end{enumerate}
Then, after having carried out the preprocessing phases for $S_\ell$ and $T_{\ell,I}$ for all $\ell\in[1,m]$ and all $\emptyset\neq I\subseteq[\ell{+}1,m]$, we can provide random-access to 
$S_1\cup\cdots\cup S_m$ in such a way that upon input of an arbitrary number $j$ it takes time 
\[
  O(\, m {\cdot} t_{\textit{acc}} + m^2 {\cdot} t_{\textit{test}} + 2^m {\cdot} t_{\textit{inv-acc}} + 2^m {\cdot }t_{\textit{lar}} \,)
\]
to output the $j$-th element that is returned by the enumeration algorithm for $S_1\cup\cdots\cup S_m$ obtained by an iterated application of Algorithm~\ref{alg:DS-UnionTrick} (starting with $A=S_1$ and $B=S_2\cup\cdots\cup S_m$).

Finally, to prove Theorem~\ref{thm:RAforUCQs}, we show the algorithms
for the different components. The quadratic-logarithmic part is due to
the $t_{\textit{lar}}$ component, and we show that it suffices for $\Largest(a)$.

\section{Implementation and Experiments}\label{sec:experiments}

In this section, we present an implementation and an experimental evaluation of the random-order enumeration algorithms presented in this paper. 
Our algorithm for random-order CQ enumeration proposed in Section~\ref{sec:cqs} is denoted as 
\ouralgCQ, the algorithm for UCQs from Section~\ref{sec:expected} is denoted \ouralgUCQ, and the algorithm for mc-UCQs from section \ref{subsec:UCQsRandomAccess} is denoted \mcUCQAlg{}.
The goal of our experiments is twofold.
First, we examine the practical execution cost of \ouralgCQ compared to the alternative of repeatedly applying a state-of-the-art sampling algorithm (without replacement) ~\cite{zhao2018random} and removing duplicates.
Second, we examine the empirical overhead of \ouralgUCQ and \mcUCQAlg{} compared to the cumulative cost of running each \ouralgCQ for each CQ separately.
We describe our implementation of \ouralgCQ,  \ouralgUCQ, and \mcUCQAlg{} in Section~\ref{sec:impl}, the experimental setup in
Section~\ref{sec:setup}, and the results in Section~\ref{sec:results}.

\subsection{Implementation}\label{sec:impl}

\ouralgCQ, \ouralgUCQ, and \mcUCQAlg{} are implemented in c++14 using the
standard library (STL), and mainly the unordered STL containers. For
instance, we use an unordered map to partition a table into the
buckets of Algorithms~\ref{alg:preprocess}
and~\ref{alg:access}. Other than STL, the implementation uses Boost to
hash complex types such as vectors.

The \ouralgCQ implementation uses a query compiler that generates c++
code for the specific CQ and database schema. Specifically, the code
is generated via templates, which are files of c++ code with
\e{placeholders}. These placeholders stand for query-specific
parameters such as the relation names, the attributes and their types,
the tree structure of the query, and its head variables. Once these
placeholders are filled in and function calls are ordered according to
the tree structure, the result is valid c++ code.

As described in Algorithm~\ref{alg:union}, \ouralgUCQ uses CQ
enumerators as black boxes with an interface of four methods:
\emph{count}, \emph{sample}, \emph{test}, and \emph{delete}.
Therefore, in addition to the counting and sampling provided by
\ouralgCQ, we implemented deletion and testing as explained
in Section~\ref{sec:expected}. The latter two require an inverted-access, which we described in Algorithm~\ref{alg:index}. The
inverted-access is compiled only when needed as part of a UCQ
enumeration, as it requires non-negligible preprocessing (to
support line~\ref{line:find_t} of Algorithm~\ref{alg:index}). Hence,
\ouralgCQ meets the four requirements when inverted-access is
activated and the shuffler capable of deletion is used. Other than
that, our implementation of Algorithm~\ref{alg:union} is fairly
straightforward.

\mcUCQAlg{} uses the underlying index \ouralgCQ{} for ran\-dom-access, testing, and inverted-access of all CQs, as well as all intersection CQs. 
We created \mcUCQAlg by using the shuffler described in \Cref{alg:shuffle} on the random-access for mcUCQs described in \Cref{subsec:UCQsRandomAccess}. Doing so requires knowing the number of answers after linear time preprocessing. The cardinality of a mcUCQ $Q_1(I) \cup \ldots \cup Q_m(I)$ is simple to compute recursively via the formula $|Q_1(I)| + |Q_2(I) \cup \cdots \cup Q_m(I)| - |Q_1(I) \cap (Q_2(I) \cup \cdots \cup Q_m(I))|$, for which we have all elements after linear time preprocessing. 
A minor difference between the implementation and the definition in \Cref{subsec:UCQsRandomAccess} instead of computing the largest answer less than or equal to our current answer and then applying inverted-access on it, we compute that index directly (using the same binary-search concept as in the proof of \Cref{thm:RAforUCQs}).

\subsection{Experimental Setup}\label{sec:setup}
We now describe the setup of our experimental study.

\subsubsection*{Algorithms} \label{sec:algs}
To the best of our knowledge, this paper is the first to suggest a
provably uniform random-order algorithm for CQ enumeration. Therefore,
we compare our \ouralgCQ to a sampling algorithm by Zhao et
al.~\cite{zhao2018random} via an implementation from their public
repository. Their algorithm generates a uniform sample, and we naively
transform it into a sampling-without-replacement algorithm by
duplicate elimination (i.e., rejecting previously encountered
answers).\footnote{The application of this approach as an enumeration
algorithm has also been discussed by Capelli and Strozecki~\cite{DBLP:journals/dam/CapelliS19}.}
Zhao et al.~\cite{zhao2018random} suggest four different
ways to initialize their algorithm, denoted \e{RS}, \e{EO}, \e{OE},
and \e{EW}. We compare our algorithm to EW as it consistently
outperformed all other methods in our experiments (see
Section~\ref{sec:other_methods} in the Appendix). We denote this
variant by \zhaoalg{EW}. This sampling algorithm is also implemented
in c++14. Hence, we consider four algorithms: \ouralgCQ,
\zhaoalg{EW}, \ouralgUCQ, and \mcUCQAlg{}.

\subsubsection*{Dataset}\label{sec:dataset}
We used the TPC-H benchmark as the database for the experiments. We
generated a database using the TPC-H \emph{dbgen} tool with a scale
factor of $\mathit{sf}=5$. The database has been instantiated once in
memory, and all experiments use the exact same database.

\subsubsection*{Queries} \label{sec:queries} 
We compare our \ouralgCQ
to \zhaoalg{EW} using the six free-connex CQs on which \zhaoalg{EW} is
implemented in the online repository. These are full-join
(projection-free) CQs over the TPC-H schema.
For lack of benchmarks, we phrased UCQs that we believed would form a
natural extension to the TPC-H queries. Also, in $Q_3$, $Q_7$, $Q_9$,
and $Q_{10}$, we added attributes from the \e{LineItem} relation to
the query in order to achieve an equivalence between set semantics and
bag semantics. The full description of our queries can be found in the
Appendix (Section~\ref{sec:appendix_queries}). Each result is the
average over three runs, except for
Figures~\ref{fig:cq_delay_full_boxplot}
and~\ref{fig:cq_delay_half_boxplot} that show a single run.

\subsubsection*{Hardware and system} \label{sec:hardware} The
experiments were executed on an Intel(R) Xeon(R) CPU 2.50GHz machine
with 768KB L1 cache, 3MB L2 cache, 30MiB L3 cache, and 496 GB of RAM,
running Ubuntu 16.04.01 LTS.  Code compilations used the \textsf{O3}
flag and no other optimization flag.

  {
    \def\chartwidth {0.25\textwidth}
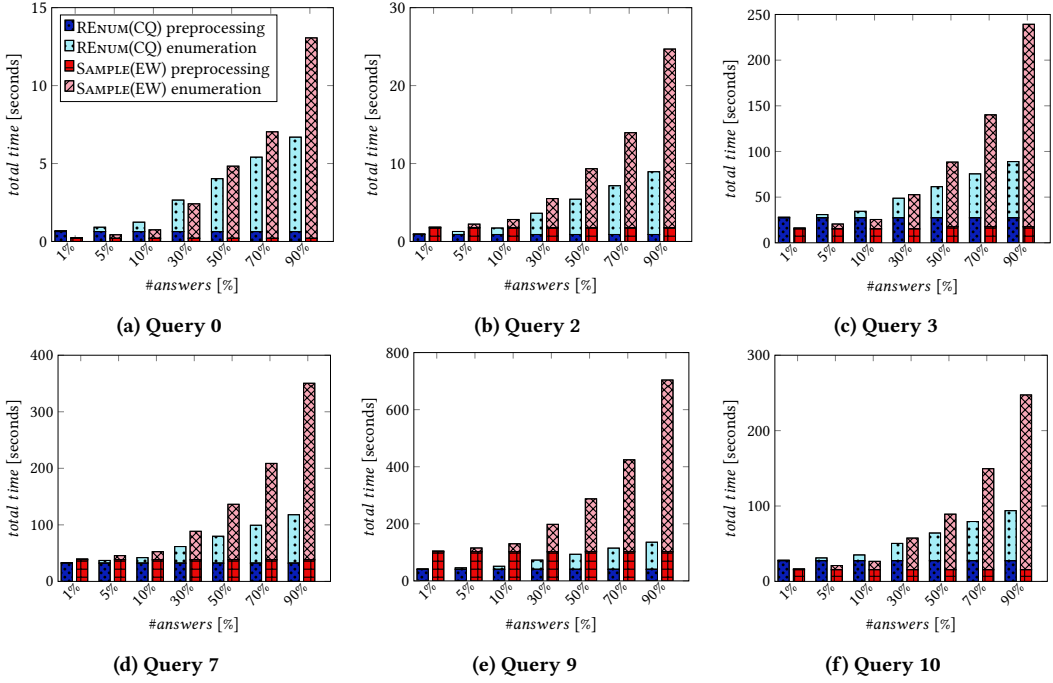
\begin{figure*}[t]
\centering
     \begin{subfigure}[b]{\chartwidth}
          \centering
          \resizebox{\linewidth}{!}{%
\begin{tikzpicture}[
  every axis/.style={ 
    ybar stacked,
    ymin = 0,ymax = 15,
    symbolic x coords={1\%, 5\%, 10\%, 30\%, 50\%, 70\%, 90\%},
    xtick={1\%, 5\%, 10\%, 30\%, 50\%, 70\%, 90\%},
    bar width=8pt,
    ylabel near ticks,
    xlabel near ticks,
    xlabel={$\#answers$ [\%]},
  ylabel={$total\;time$ [seconds}],
  label style={font=\LARGE},
  tick label style={font=\LARGE},
  x tick label style={rotate=45,anchor=east}
  }
]

\begin{axis}[bar shift=-10pt,hide axis,legend pos=north west,legend style={font=\LARGE},legend cell align={left}]
\addplot+[black, fill=dark_blue, postaction={pattern=crosshatch dots}] coordinates
{(1\%,0.63365) (5\%,0.63365) (10\%,0.63365) (30\%,0.63365) (50\%,0.63365) (70\%,0.63365) (90\%,0.63365)};
\addlegendentry{\ouralgCQ{} preprocessing}

\addplot+[black, fill=light_blue, postaction={pattern=dots}] coordinates
{(1\%,0.05733) (5\%,0.28) (10\%,0.603) (30\%,2.021667) (50\%,3.3933) (70\%,4.77467) (90\%, 6.05867)};
\addlegendentry{\ouralgCQ{} enumeration}

\addplot+[black, fill=dark_red, postaction={pattern=horizontal lines}] coordinates {(1\%,0)};
\addlegendentry{\zhaoalg{EW} preprocessing}

\addplot+[black, fill=light_red, postaction={pattern=north east lines}] coordinates {(1\%,0)};
\addlegendentry{\zhaoalg{EW} enumeration}
\end{axis}

\begin{axis}[bar shift=1pt]
\addplot+[black, fill=dark_red, postaction={pattern=grid}] coordinates
{(1\%,0.20233) (5\%,0.20233) (10\%,0.20233) (30\%,0.20233) (50\%,0.20233) (70\%,0.20233) (90\%,0.20233)};

\addplot+[black, fill=light_red, postaction={pattern=crosshatch}] coordinates
{(1\%,0.04733) (5\%,0.22833) (10\%,0.54533) (30\%, 2.21267) (50\%,4.629) (70\%,6.83367) (90\%,12.85967)};
\end{axis}

\end{tikzpicture}
}
          \caption{Query 0}
          \label{fig:q0}
     \end{subfigure}\quad
     \begin{subfigure}[b]{\chartwidth}
          \centering
          \resizebox{\linewidth}{!}{%
\begin{tikzpicture}[
  every axis/.style={ 
    ybar stacked,
    ymin = 0,ymax = 30,
    symbolic x coords={1\%, 5\%, 10\%, 30\%, 50\%, 70\%, 90\%},
    xtick={1\%, 5\%, 10\%, 30\%, 50\%, 70\%, 90\%},
    bar width=8pt,
    ylabel near ticks,
    xlabel near ticks,
    xlabel={$\#answers$ [\%]},
  ylabel={$total\;time$ [seconds}],
  label style={font=\LARGE},
  tick label style={font=\LARGE},
  x tick label style={rotate=45,anchor=east}
  }
]

\begin{axis}[bar shift=-10pt,hide axis]
\addplot+[black, fill=dark_blue, postaction={pattern=crosshatch dots}] coordinates
{(1\%,0.8944565) (5\%,0.8944565) (10\%,0.8944565) (30\%,0.8944565) (50\%,0.8944565) (70\%,0.8944565) (90\%,0.8944565)};
\addplot+[black, fill=light_blue, postaction={pattern=dots}] coordinates
{(1\%,0.076) (5\%,0.38733) (10\%,0.83167) (30\%,2.72433) (50\%,4.52567) (70\%,6.258) (90\%,8.05566)};
\end{axis}

\begin{axis}[bar shift=1pt]
\addplot+[black, fill=dark_red, postaction={pattern=grid}] coordinates
{(1\%,1.75275) (5\%,1.75275) (10\%,1.75275) (30\%,1.75275) (50\%,1.75275) (70\%,1.75275) (90\%,1.75275)};
\addplot+[black, fill=light_red, postaction={pattern=crosshatch}] coordinates
{(1\%,0.092) (5\%,0.472) (10\%,1.069) (30\%,3.7493) (50\%,7.58233) (70\%,12.20233) (90\%,22.93566)};
\end{axis}

\end{tikzpicture}
}
          \caption{Query 2}
          \label{fig:q2}
     \end{subfigure}\quad
     \begin{subfigure}[b]{\chartwidth}
          \centering
          \resizebox{\linewidth}{!}{%
\begin{tikzpicture}[
  every axis/.style={ 
    ybar stacked,
    ymin = 0,ymax = 250,
    symbolic x coords={1\%,5\%,10\%,30\%,50\%,70\%,90\%},
    xtick={1\%,5\%,10\%,30\%,50\%,70\%,90\%},
  bar width=8pt,
  ylabel near ticks,
  xlabel near ticks,
  xlabel={$\#answers$ [\%]},
  ylabel={$total\;time$ [seconds}],
  label style={font=\LARGE},
  tick label style={font=\LARGE}
  ,x tick label style={rotate=45,anchor=east}
  },
]

\begin{axis}[bar shift=-10pt,hide axis]
\addplot+[black, fill=dark_blue, postaction={pattern=crosshatch dots}] coordinates
{(1\%,27.35423) (5\%,27.35423) (10\%,27.35423) (30\%,27.35423) (50\%,27.35423) (70\%,27.35423) (90\%,27.35423)};
\addplot+[black, fill=light_blue, postaction={pattern=dots}] coordinates
{(1\%,0.675) (5\%,3.542) (10\%,7.134) (30\%,21.45467) (50\%,34.089) (70\%,48.14733) (90\%,61.551)}; %
\end{axis}

\begin{axis}[bar shift=1pt]
\addplot+[black, fill=dark_red, postaction={pattern=grid}] coordinates
{(1\%,15.18867) (5\%,15.18867) (10\%,15.18867) (30\%,15.18867) (50\%,17.517) (70\%,17.517) (90\%,17.517)};
\addplot+[black, fill=light_red, postaction={pattern=crosshatch}] coordinates
{(1\%,1.08033) (5\%,5.389) (10\%,10.34467) (30\%,37.47967) (50\%,70.85967) (70\%,122.64067) (90\%,221.76033)};
\end{axis}

\end{tikzpicture}
}
          \caption{Query 3}
          \label{fig:q3}
        \end{subfigure}
        
     \begin{subfigure}[b]{\chartwidth}
          \centering
          \resizebox{\linewidth}{!}{
\begin{tikzpicture}[
  every axis/.style={ 
    ybar stacked,
    ymin = 0,ymax = 400,
    symbolic x coords={1\%,5\%,10\%,30\%,50\%,70\%,90\%}, 
    xtick={1\%,5\%,10\%,30\%,50\%,70\%,90\%},
  bar width=8pt,
  ylabel near ticks,
  xlabel near ticks,
  xlabel={$\#answers$ [\%]},
  ylabel={$total\;time$ [seconds}],
  label style={font=\LARGE},
  tick label style={font=\LARGE}
  ,x tick label style={rotate=45,anchor=east}
  },
]

\begin{axis}[bar shift=-10pt,hide axis]
\addplot+[black, fill=dark_blue, postaction={pattern=crosshatch dots}] coordinates
{(1\%,32.285) (5\%,32.285) (10\%,32.285) (30\%,32.285) (50\%,32.285) (70\%,32.285) (90\%,32.285)}; %
\addplot+[black, fill=light_blue, postaction={pattern=dots}] coordinates
{(1\%,0.88867) (5\%,4.805) (10\%,9.7533) (30\%,29.3493) (50\%,47.556) (70\%,67.004667) (90\%,85.63367)};
\end{axis}

\begin{axis}[bar shift=1pt]
\addplot+[black, fill=dark_red, postaction={pattern=grid}] coordinates
{(1\%,38.23233) (5\%,38.23233) (10\%,38.23233) (30\%,38.23233) (50\%,38.23233) (70\%,38.23233) (90\%,38.23233)};
\addplot+[black, fill=light_red, postaction={pattern=crosshatch}] coordinates
{(1\%,1.48867) (5\%,7.38067) (10\%,14.332) (30\%,50.40167) (50\%,98.1633) (70\%,170.45567) (90\%,312.18067)};
\end{axis}

\end{tikzpicture}
}
          \caption{Query 7}
          \label{fig:q7}
     \end{subfigure}\quad
     \begin{subfigure}[b]{\chartwidth}
          \centering
          \resizebox{\linewidth}{!}{%
\begin{tikzpicture}[
  every axis/.style={ 
    ybar stacked,
    ymin = 0,ymax = 800,
    symbolic x coords={1\%,5\%,10\%,30\%,50\%,70\%,90\%}, 
    xtick={1\%,5\%,10\%,30\%,50\%,70\%,90\%},
  bar width=8pt,
  ylabel near ticks,
  xlabel near ticks,
  xlabel={$\#answers$ [\%]},
  ylabel={$total\;time$ [seconds}],
  label style={font=\LARGE},
  tick label style={font=\LARGE}
  ,x tick label style={rotate=45,anchor=east}
  },
]

\begin{axis}[bar shift=-10pt,hide axis]
\addplot+[black, fill=dark_blue, postaction={pattern=crosshatch dots}] coordinates
{(1\%,40.5219) (5\%,40.5219) (10\%,40.5219) (30\%,40.5219) (50\%,40.5219) (70\%,40.5219) (90\%,40.5219)};
\addplot+[black, fill=light_blue, postaction={pattern=dots}] coordinates
{(1\%,1.010667) (5\%,5.28833) (10\%,10.68567) (30\%,32.36533) (50\%,52.718) (70\%,74.29467) (90\%,94.95167)};
\end{axis}

\begin{axis}[bar shift=1pt]
\addplot+[black, fill=dark_red, postaction={pattern=grid}] coordinates
{(1\%,101.901) (5\%,101.901) (10\%,101.901) (30\%,101.901) (50\%,101.901) (70\%,101.901) (90\%,101.901)};
\addplot+[black, fill=light_red, postaction={pattern=crosshatch}] coordinates
{(1\%,2.72066) (5\%,13.46667) (10\%,28.16267) (30\%,96.23133) (50\%,185.639) (70\%,322.401) (90\%,601.94833)};
\end{axis}

\end{tikzpicture}
}
          \caption{Query 9}
          \label{fig:q9}
     \end{subfigure}\quad
     \begin{subfigure}[b]{\chartwidth}
          \centering
          \resizebox{\linewidth}{!}{
\begin{tikzpicture}[
  every axis/.style={ 
    ybar stacked,
    ymin = 0,ymax = 300,
    symbolic x coords={1\%,5\%,10\%,30\%,50\%,70\%,90\%},
    xtick={1\%,5\%,10\%,30\%,50\%,70\%,90\%},
  bar width=8pt,
  ylabel near ticks,
  xlabel near ticks,
  xlabel={$\#answers$ [\%]},
  ylabel={$total\;time$ [seconds}],
  label style={font=\LARGE},
  tick label style={font=\LARGE}
  ,x tick label style={rotate=45,anchor=east}
  },
]

\begin{axis}[bar shift=-10pt,hide axis]
\addplot+[black, fill=dark_blue, postaction={pattern=crosshatch dots}] coordinates
{(1\%,27.4471) (5\%,27.4471) (10\%,27.4471) (30\%,27.4471) (50\%,27.4471) (70\%,27.4471) (90\%,27.4471)};
\addplot+[black, fill=light_blue, postaction={pattern=dots}] coordinates
{(1\%,0.69033) (5\%,3.74367) (10\%,7.683) (30\%,22.948) (50\%,36.74167) (70\%,51.97567) (90\%,66.39467)};
\end{axis}

\begin{axis}[bar shift=1pt]
\addplot+[black, fill=dark_red, postaction={pattern=grid}] coordinates
{(1\%,15.49433) (5\%,15.49433) (10\%,15.49433) (30\%,15.49433) (50\%,15.49433) (70\%,15.49433) (90\%,15.49433)};
\addplot+[black, fill=light_red, postaction={pattern=crosshatch}] coordinates
{(1\%,1.166) (5\%,5.59033) (10\%,11.12633) (30\%,41.98467) (50\%,73.701) (70\%,134.005) (90\%,231.826)};
\end{axis}

\end{tikzpicture}
}
          \caption{Query 10}
          \label{fig:q10}
     \end{subfigure}
     \vspace{-1em}
     \caption{Total enumeration time of CQs when requesting different
       percentages of answers. In each bar, the bottom (darker) part
       refers to the preprocessing phase and the top (lighter) part to
       the enumeration
       phase.}
     \label{fig:cq_total_time}
     \vspace{-1em}
   \end{figure*}
   }
  
\subsection{Experimental Results}\label{sec:results}

We now describe the results of our experimentation with CQs and UCQs.
The CQ experiments analyze \ouralgCQ in terms of the total enumeration
time (Section~\ref{sec:cq_total_time}) and delay
(Section~\ref{sec:cq_delay}), while the UCQ experiments analyze
\ouralgUCQ and \mcUCQAlg{} in terms of the total enumeration time, as well as the rejection rate of \ouralgUCQ{}
(Section~\ref{sec:ucq_results}).
We omit from
all preprocessing times the portion devoted to reading the relations.

    \subsubsection{CQ running time} \label{sec:cq_total_time}

To characterize the total enumeration time of \ouralgCQ, we compare it
to that of \zhaoalg{EW} for the TPC-H CQs. In the experiment, we task
each algorithm with enumerating $k$ distinct answers for increasing
values of $k$. The different values of $k$ were chosen as a percentage
of the query results ({1\%, 5\%, 10\%, 30\%, 50\%, 70\%,
  90\%}). For each 
task, we measure the total enumeration
time, that is, the time elapsed from the beginning of the
preprocessing phase to when $k$ distinct answers were supplied. The
results of this experiment are presented in \Cref{fig:cq_total_time}
with a chart per query. 
The results indicate that, as $k$ grows, the total enumeration time of
\zhaoalg{EW} grows more rapidly in comparison to \ouralgCQ. Generally,
the total time of \ouralgCQ increases slower as it does not reject
answers. Hence, \zhaoalg{EW} seems better or comparable for smaller
$k$ values, but is consistently outperformed by \ouralgCQ for larger
values of $k$. This is especially true when preprocessing time becomes
negligible in comparison to the time it takes to enumerate $k$
distinct answers.
\ouralgCQ performs better, relative to \zhaoalg{EW}, on queries with
more relations ($Q_2$ ,$Q_7$, $Q_9$) than ones with fewer relations.

{
\begin{figure}[b]
\centering
     \begin{subfigure}[b]{0.15\textwidth}
          \centering
          \resizebox{\linewidth}{!}{\includegraphics{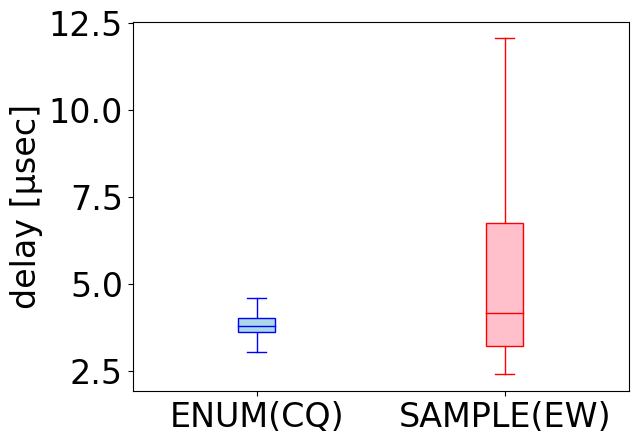}}
          \caption{Query 0}
          \label{fig:q0_full_boxplot}
     \end{subfigure}
     \begin{subfigure}[b]{0.15\textwidth}
          \centering
          \resizebox{\linewidth}{!}{\includegraphics{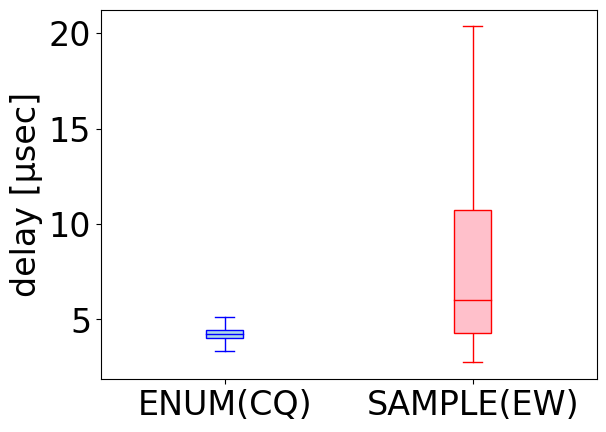}}
          \caption{Query 2}
          \label{fig:q2_full_boxplot}
     \end{subfigure}
     \begin{subfigure}[b]{0.15\textwidth}
          \centering
          \resizebox{\linewidth}{!}{\includegraphics{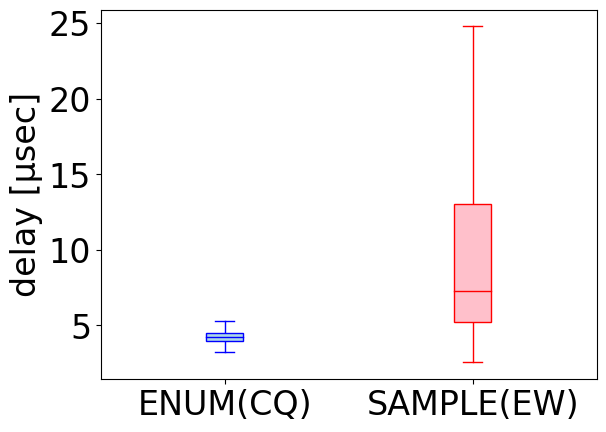}}
          \caption{Query 3}
          \label{fig:q3_full_boxplot}
     \end{subfigure}
     \begin{subfigure}[b]{0.15\textwidth}
          \centering
          \resizebox{\linewidth}{!}{\includegraphics{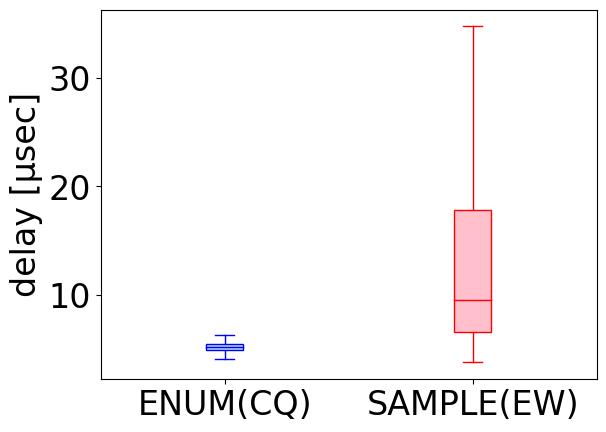}}
          \caption{Query 7}
          \label{fig:q7_full_boxplot}
     \end{subfigure}
     \begin{subfigure}[b]{0.15\textwidth}
          \centering
          \resizebox{\linewidth}{!}{\includegraphics{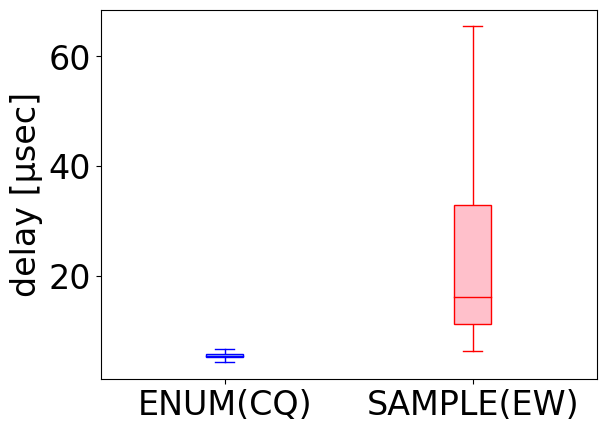}}
          \caption{Query 9}
          \label{fig:q9_full_boxplot}
     \end{subfigure}
     \begin{subfigure}[b]{0.15\textwidth}
          \centering
          \resizebox{\linewidth}{!}{\includegraphics{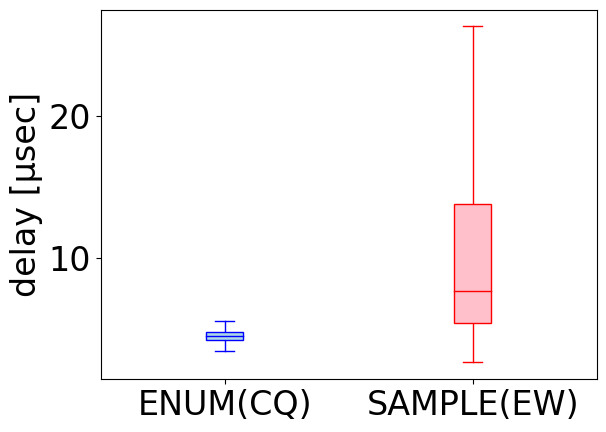}}
          \caption{Query 10}
          \label{fig:q10_full_boxplot}
     \end{subfigure}
     \vskip-1em
     \caption{The delay in a full enumeration.}
     \label{fig:cq_delay_full_boxplot}
   \end{figure}
 }

 {
\begin{figure}[b]
\centering
     \begin{subfigure}[b]{0.15\textwidth}
          \centering
          \resizebox{\linewidth}{!}{\includegraphics{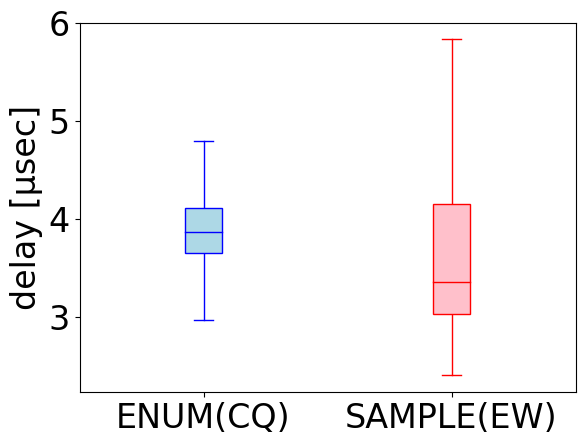}}
          \caption{Query 0}
          \label{fig:q0_half_boxplot}
     \end{subfigure}
     \begin{subfigure}[b]{0.15\textwidth}
          \centering
          \resizebox{\linewidth}{!}{\includegraphics{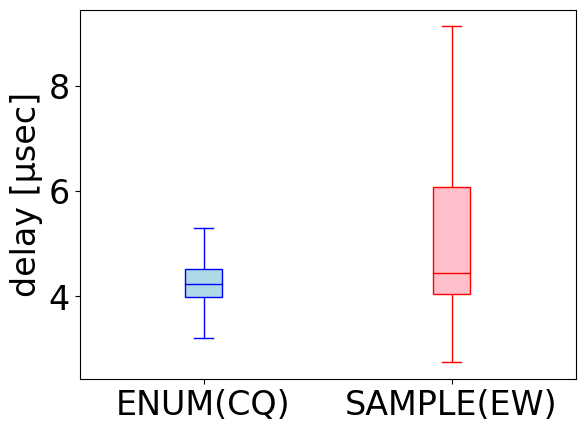}}
          \caption{Query 2}
          \label{fig:q2_half_boxplot}
     \end{subfigure}
     \begin{subfigure}[b]{0.15\textwidth}
          \centering
          \resizebox{\linewidth}{!}{\includegraphics{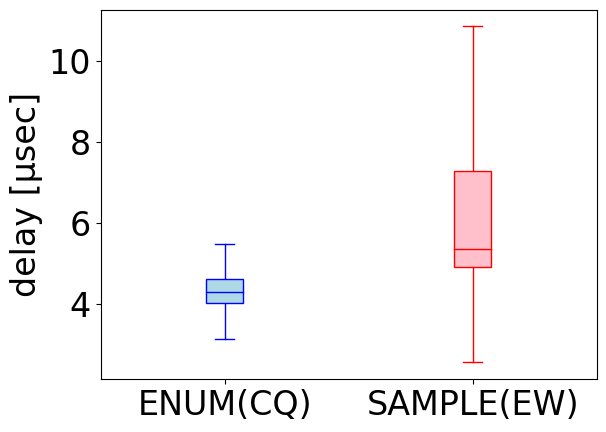}}
          \caption{Query 3}
          \label{fig:q3_half_boxplot}
     \end{subfigure}
     \begin{subfigure}[b]{0.15\textwidth}
          \centering
          \resizebox{\linewidth}{!}{\includegraphics{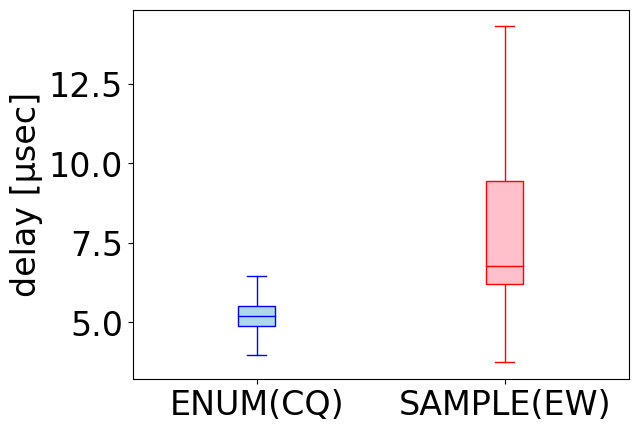}}
          \caption{Query 7}
          \label{fig:q7_half_boxplot}
     \end{subfigure}
     \begin{subfigure}[b]{0.15\textwidth}
          \centering
          \resizebox{\linewidth}{!}{\includegraphics{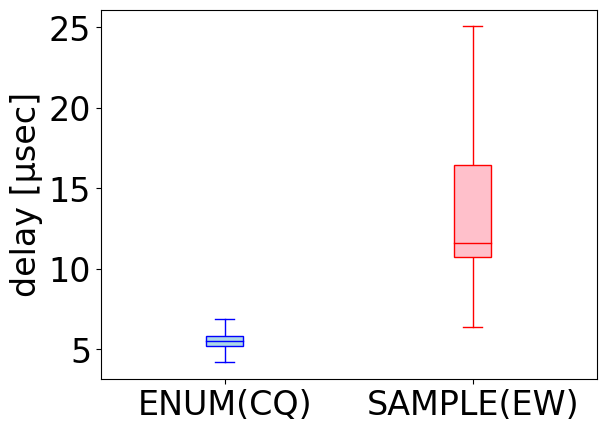}}
          \caption{Query 9}
          \label{fig:q9_half_boxplot}
     \end{subfigure}
     \begin{subfigure}[b]{0.15\textwidth}
          \centering
          \resizebox{\linewidth}{!}{\includegraphics{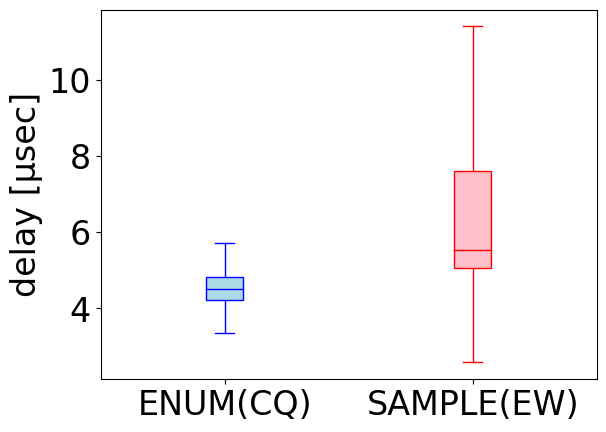}}
          \caption{Query 10}
          \label{fig:q10_half_boxplot}
     \end{subfigure}
     \vskip-1em
     \caption{The delay when enumerating 50\% of the answers.}

     \label{fig:cq_delay_half_boxplot}
 \end{figure}
}

\subsubsection{CQ delay analysis} \label{sec:cq_delay} To examine the
delay of \ouralgCQ{} and \zhaoalg{EW}, we record the delay of each
answer and depict it in box-and-whisker plots. Each query has two box
plots: enumeration of all answers (\Cref{fig:cq_delay_full_boxplot})
and enumeration of 50\% of the answers
(\Cref{fig:cq_delay_half_boxplot}). Outliers that fell outside the
whiskers are not shown, since some are several orders of magnitude
larger than the median. Further information regarding the results of
this experiment, as well as the outliers dropped, is in the Appendix
(Section~\ref{sec:delay_tables}). We can see that in a full
enumeration, \ouralgCQ always shows a lower median value, smaller
variation, and a smaller interquartile range (IQR). Smaller IQR and
whiskers show that half of the delay samples fall within a smaller
range, meaning that the delay is more stable and predictable. When
enumerating 50\% of the answers, the variation and IQR remain smaller
across all queries. However, in $Q_0$ we see that \zhaoalg{EW}
actually exhibits a smaller median. In addition, \ouralgCQ usually
shows better results on larger queries, in comparison to
\zhaoalg{EW}. For instance, \zhaoalg{EW} has a better median in $Q_0$
than in its larger counterpart $Q_2$.

{
\begin{figure}[t]
\centering 
\begin{subfigure}{0.25\textwidth}
    \centering
    \resizebox{\linewidth}{!}{
\def \UCQone{\text{$Q_7^S \cup Q_7^C$}}
\def \UCQtwo{\text{$Q_2^N \cup Q_2^P \cup Q_2^S$}}
\def \UCQthree{\text{$Q_A \cup Q_E$}}

\begin{tikzpicture}[
  every axis/.style={
    ybar stacked,
    ymin = 0,ymax = 70,
    symbolic x coords={\UCQthree,\UCQone,\UCQtwo},
    xtick={\UCQthree,\UCQone,\UCQtwo},
    bar width=8pt,
    ylabel near ticks,
    ylabel={$total\;time$ [seconds}],
    label style={font=\Large},
    tick label style={font=\Large},
  }
]

\begin{axis}[bar shift=-10pt, hide axis]
\addplot+[black, fill=dark_green, postaction={pattern=crosshatch dots}] coordinates {(\UCQone,4.841567) (\UCQtwo,0) (\UCQthree,0)};

\addplot+[black, fill=dark_green, postaction={pattern=crosshatch dots}] coordinates {(\UCQone,7.455633) (\UCQtwo,0) (\UCQthree,0)};

\addplot+[black, fill=light_green, postaction={pattern=dots}] coordinates {(\UCQone,2.49933) (\UCQtwo,0) (\UCQthree,0)};

\addplot+[black, fill=light_green, postaction={pattern=dots}] coordinates {(\UCQone,2.45767) (\UCQtwo,0) (\UCQthree,0)};

\addplot+[black, fill=dark_green, postaction={pattern=crosshatch dots}] coordinates {(\UCQone,0) (\UCQtwo,0.88117866666667) (\UCQthree,0)}; %
\addplot+[black, fill=dark_green, postaction={pattern=crosshatch dots}] coordinates {(\UCQone,0) (\UCQtwo,0.610207) (\UCQthree,0)}; %
\addplot+[black, fill=dark_green, postaction={pattern=crosshatch dots}] coordinates {(\UCQone,0) (\UCQtwo,0.67323133333333) (\UCQthree,0)}; %
\addplot+[black, fill=light_green, postaction={pattern=dots}] coordinates {(\UCQone,0) (\UCQtwo,8.5896666666667) (\UCQthree,0)}; %
\addplot+[black, fill=light_green, postaction={pattern=dots}] coordinates {(\UCQone,0) (\UCQtwo,3.8776666666667) (\UCQthree,0)}; %
\addplot+[black, fill=light_green, postaction={pattern=dots}] coordinates {(\UCQone,0) (\UCQtwo,3.8403333333333) (\UCQthree,0)}; %

\addplot+[black, fill=dark_green, postaction={pattern=crosshatch dots}] coordinates {(\UCQone,0) (\UCQtwo,0) (\UCQthree, 3.67774)}; %
\addplot+[black, fill=dark_green, postaction={pattern=crosshatch dots}] coordinates {(\UCQone,0) (\UCQtwo,0) (\UCQthree,3.7115233333333)}; %
\addplot+[black, fill=light_green, postaction={pattern=dots}] coordinates {(\UCQone,0) (\UCQtwo,0) (\UCQthree,2.2356666666667)}; %
\addplot+[black, fill=light_green, postaction={pattern=dots}] coordinates {(\UCQone,0) (\UCQtwo,0) (\UCQthree,2.2663333333333)}; %
\end{axis}

\begin{axis}[bar shift=0pt, hide axis]
\addplot+[black, fill=dark_blue, postaction={pattern=grid}] coordinates {(\UCQone,13.4404) (\UCQtwo,0) (\UCQthree,0)}; %
\addplot+[black, fill=light_blue, postaction={pattern=crosshatch}] coordinates {(\UCQone,10.083666666667) (\UCQtwo,0) (\UCQthree,0)}; %

\addplot+[black, fill=dark_blue, postaction={pattern=grid}] coordinates {(\UCQone,0) (\UCQtwo,5.95336) (\UCQthree,0)}; %
\addplot+[black, fill=light_blue, postaction={pattern=crosshatch}] coordinates {(\UCQone,0) (\UCQtwo,40.142333333333) (\UCQthree,0)}; %

\addplot+[black, fill=dark_blue, postaction={pattern=grid}] coordinates {(\UCQone,0) (\UCQtwo,0) (\UCQthree,8.8582633333333)}; %
\addplot+[black, fill=light_blue, postaction={pattern=crosshatch}] coordinates {(\UCQone,0) (\UCQtwo,0) (\UCQthree,8.112)}; %
\end{axis}

\begin{axis}[bar shift=10pt]
\addplot+[black, fill=dark_purple, postaction={pattern=grid}] coordinates {
(\UCQone,15.5203)
(\UCQtwo,0)
(\UCQthree,0)
}; %

\addplot+[black, fill=light_purple, postaction={pattern=crosshatch}] coordinates {
(\UCQone,6.76267)
(\UCQtwo,0)
(\UCQthree,0)
}; %

\addplot+[black, fill=dark_purple, postaction={pattern=grid}] coordinates {
(\UCQone,0)
(\UCQtwo,9.98706)
(\UCQthree,0)
}; %

\addplot+[black, fill=light_purple, postaction={pattern=crosshatch}] coordinates {
(\UCQone,0)
(\UCQtwo,18)
(\UCQthree,0)
};

\addplot+[draw=none,fill=white] coordinates {
(\UCQone,0)
(\UCQtwo,8)
(\UCQthree,0)
};

\addplot+[black, fill=light_purple, postaction={pattern=crosshatch}] coordinates {
(\UCQone,0)
(\UCQtwo,18)
(\UCQthree,0)
};

\addplot+[black, fill=dark_purple, postaction={pattern=grid}] coordinates {(\UCQone,0) (\UCQtwo,0) (\UCQthree,9.3179)}; %
\addplot+[black, fill=light_purple, postaction={pattern=crosshatch}] coordinates {(\UCQone,0) (\UCQtwo,0) (\UCQthree,5.497)}; %

\end{axis}  

\begin{axis}[bar shift=0pt,hide axis,legend pos=north west,legend style={font=\Large},legend cell align={left}]
\addplot+[black, fill=dark_green, postaction={pattern=crosshatch dots}] coordinates {(\UCQone,0) (\UCQtwo,0) (\UCQthree,0)};
\addlegendentry{\ouralgCQ{} preprocessing}

\addplot+[black, fill=light_green, postaction={pattern=dots}] coordinates {(\UCQone,0) (\UCQtwo,0) (\UCQthree,0)};
\addlegendentry{\ouralgCQ{} enumeration}

\addplot+[black, fill=dark_blue, postaction={pattern=grid}] coordinates {(\UCQone,0) (\UCQtwo,0) (\UCQthree,0)}; %
\addlegendentry{\ouralgUCQ{} preprocessing}

\addplot+[black, fill=light_blue, postaction={pattern=crosshatch}] coordinates {(\UCQone,0) (\UCQtwo,0) (\UCQthree,0)}; %
\addlegendentry{\ouralgUCQ{} enumeration}

\addplot+[black, fill=dark_purple, postaction={pattern=grid}] coordinates {(\UCQone,0) (\UCQtwo,0) (\UCQthree,0)}; %
\addlegendentry{\mcUCQAlg{} preprocessing}

\addplot+[black, fill=light_purple, postaction={pattern=crosshatch}] coordinates {(\UCQone,0) (\UCQtwo,0) (\UCQthree,0)}; %
\addlegendentry{\mcUCQAlg{} enumeration}
\end{axis}

\coordinate (dots) at (6.65,2.7);
\coordinate (num) at (6.3,5.3);
\node at (dots) {$\vdots$};
\node at (num) {\Large $493.63$};
\draw[->] ($(num) + (0.1,-0.2)$) -- (6.65,4.45);

\end{tikzpicture}
}
    \caption{Total time.}
    \label{fig:all_ucqs_total_time}
\end{subfigure}
\begin{subfigure}{0.22\textwidth}
    \centering
    \resizebox{\linewidth}{!}{%
\begin{tikzpicture}[
  every axis/.style={ 
    ybar stacked,
    ymin = 0,ymax = 25,
    symbolic x coords={1\%,5\%,10\%,30\%,50\%,70\%,90\%, 100\%},
    xtick={1\%,5\%,10\%,30\%,50\%,70\%,90\%,100\%},
    bar width=5pt,
    ylabel near ticks,
    ylabel={$total\;time$ [seconds}],
    label style={font=\Large},
    tick label style={font=\Large},
    x tick label style={rotate=45,anchor=east}
  }
]

\begin{axis}[bar shift=-12pt] %
\addplot+[black, fill=dark_green, postaction={pattern=crosshatch dots}] coordinates {(1\%,4.9348433333333) (5\%,4.9348433333333) (10\%,4.9348433333333) (30\%,4.9348433333333) (50\%,4.9348433333333) (70\%,4.9348433333333) (90\%,4.9348433333333) (100\%,4.9348433333333)};

\addplot+[black, fill=dark_green, postaction={pattern=crosshatch dots}] coordinates {(1\%,7.4491966666667) (5\%,7.4491966666667) (10\%,7.4491966666667) (30\%,7.4491966666667) (50\%,7.4491966666667) (70\%,7.4491966666667) (90\%,7.4491966666667) (100\%,7.4491966666667)};

\addplot+[black, fill=light_green, postaction={pattern=dots}] coordinates {(1\%,0.024333333333333) (5\%,0.122) (10\%,0.247) (30\%,0.78933333333333) (50\%,1.3383333333333) (70\%,1.8336666666667) (90\%,2.3723333333333) (100\%,2.6396666666667)};

\addplot+[black, fill=light_green, postaction={pattern=dots}] coordinates {(1\%,0.024) (5\%,0.114) (10\%,0.22966666666667) (30\%,0.737) (50\%,1.27) (70\%,1.7403333333333) (90\%,2.232) (100\%,2.465)};
\end{axis}

\begin{axis}[bar shift=-5pt, hide axis] %
\addplot+[black, fill=dark_blue, postaction={pattern=grid}] coordinates {(1\%,13.4597) (5\%,13.4597) (10\%,13.4597) (30\%,13.4597) (50\%,13.4597) (70\%,13.4597) (90\%,13.4597) (100\%,13.4597)};
\addplot+[black, fill=light_blue, postaction={pattern=crosshatch}] coordinates {(1\%,0.158) (5\%,0.523) (10\%,1.0026666666667) (30\%,3.23) (50\%,4.6766666666667) (70\%,7.222) (90\%,8.522) (100\%,9.5466666666667)};
\end{axis}

\begin{axis}[bar shift=2pt, hide axis]
\addplot+[black, fill=dark_purple, postaction={pattern=grid}] coordinates {(1\%,15.5203) (5\%,15.5203) (10\%,15.5203) (30\%,15.5203) (50\%,15.5203) (70\%,15.5203) (90\%,15.5203) (100\%,15.5203)};
\addplot+[black, fill=light_purple, postaction={pattern=crosshatch}] coordinates {(1\%,0.112667) (5\%,0.327) (10\%,0.659667) (30\%,2.1) (50\%,3.48467) (70\%,4.789) (90\%,6.11567) (100\%,6.76267)};
\end{axis}

\end{tikzpicture}
}
    \caption{Varying percentage.}
    \label{fig:ucq1_total_time}
  \end{subfigure}
  \vskip-1em
\caption{The total time of UCQs with \ouralgUCQ vs.~the total time of CQs comprising the union with \ouralgCQ.}
\end{figure}
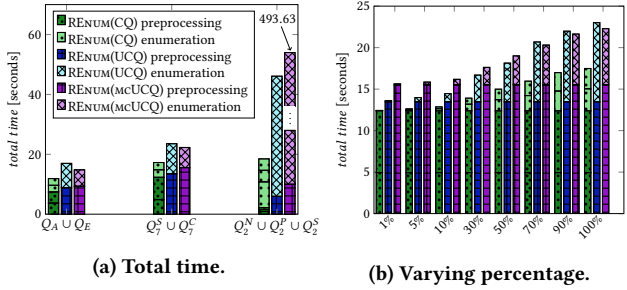
}

\subsubsection{UCQ enumeration} \label{sec:ucq_results} This section
analyzes the total enumeration time of \ouralgUCQ and \mcUCQAlg, as
well as the time spent on rejection of \ouralgUCQ in three
experiments, shown in Figures~\ref{fig:all_ucqs_total_time},
\ref{fig:ucq1_total_time}, and~\ref{fig:ucq1_rejection_vs_answers},
respectively.  The
1st experiment measures the length of a full enumeration (with
\ouralgUCQ or \mcUCQAlg) in three UCQs, while the
second focuses on one UCQ and measures the total time of both UCQ
algorithms as it varies when producing a different portion of the
answers (as in Section~\ref{sec:cq_total_time}). In both experiments,
we compare \ouralgUCQ and \mcUCQAlg to the cumulative running time of
\ouralgCQ on the CQs comprising the union. We stress that running
\ouralgCQ on the independent CQs is \e{not} an alternative to an
actual UCQ enumeration---it produces duplicates and does not provide a
uniform random order. We perform this comparison to measure the
overhead of the UCQ algorithms over \ouralgCQ. The 
3rd experiment examines the time that \ouralgUCQ spends on producing
rejected answers during a single run. It shows how this time changes
along the course of a full enumeration.

   The difference in preprocessing time between \ouralgCQ and
   \ouralgUCQ is that for the latter, we need to build an index
   that supports \Cref{line:find_t} in \Cref{alg:union}.
   \Cref{fig:all_ucqs_total_time} shows that this difference can be
   quite small, as seen in $Q_7^S \cup Q_7^C$ and $Q_A \cup Q_E$. 
   Meanwhile, the preprocessing of \mcUCQAlg adds to that of \ouralgUCQ the need to preprocess CQs defined by intersection of CQs from the union. Hence, \mcUCQAlg always has the largest preprocessing time. 
   Nevertheless, we see that the difference in the enumeration phase
   is more significant for both algorithms.
   
   The slowdown of \ouralgUCQ{} compared to \ouralgCQ{} is mostly
   attributed to the effort to avoid duplicates by multiple CQs.  A
   union of $m$ CQs calls the inverted-access $m{-}1$ times per
   answer.  Also, the enumeration phase is slowed down by the deletion
   mechanism and rejections.  \Cref{fig:all_ucqs_total_time} also
   shows that the slowdown between \ouralgCQ{} and \ouralgUCQ{}
   depends largely on the intersection size.
   $Q_2^N \cup Q_2^P \cup Q_2^S$ has a  large intersection and
   $Q_A \cup Q_E$ has no intersection at all.  In general, two
   disjoint queries will be much faster than two identical queries
   because for two identical queries the algorithm will reject half of
   the answers on average.
   
   \Cref{fig:all_ucqs_total_time} demonstrates that the difference in
   running time between \mcUCQAlg{} and \ouralgUCQ{} depends on the
   number of CQs in the union.  For two CQs,
   \mcUCQAlg{} outperforms \ouralgUCQ{}.  $Q_A \cup Q_E$ is a disjoint
   union, while $Q_7^S \cup Q_7^C$ is not.  Both algorithms benefit
   from a disjoint union, but \mcUCQAlg{} maintains its lead.  If the
   union is disjoint, \cref{line:computek} of
   \Cref{alg:AccessUnionOfTwo} will never be called, so the running
   time of the inverted-access is saved.  In \ouralgUCQ{}, a disjoint
   union causes no rejections, so the delay is also guaranteed
   log-time (not only in expectation). However, \ouralgUCQ{} still
   tests membership in the other queries (as we do not know in advance
   that the union is disjoint), so it is more costly.  \mcUCQAlg{} on $Q_2^N \cup Q_2^P \cup Q_2^S$ suffers from a
   larger number of CQs in the union. As the delay depends
   exponentially on the number of CQs in the union, this has a
   dramatic effect.
   
   \Cref{fig:ucq1_total_time} shows the middle column of \Cref{fig:all_ucqs_total_time} as it changes during the course of enumeration. It shows that the increase in total delay is rather steady in both UCQ algorithms, and that \mcUCQAlg{} starts being preferable over \ouralgUCQ{} when producing about 60\% of the answers or more.

   Finally, \Cref{fig:ucq1_rejection_vs_answers} shows that the time
   \ouralgUCQ spends on producing rejected answers decays over
   time. A possible explanation is that the number of answers that belong to
   multiple CQs (shared answers) drops faster than that of
   non-shared answers, for two reasons. First, shared answers
   have a higher
   probability of being selected. Second, when a non-shared answer is
   selected, it is deleted everywhere, while a shared answer
   may become non-shared.

  \subsubsection{Conclusions}
  Our experimental study indicates that the merits of \ouralgCQ are not only in its complexity and
  statistical guarantees---a fairly simple implementation of it features a significant improvement in practical performance compared to the state-of-the-art approach. 
  Moreover, the overhead of \ouralgUCQ is non-negligible. While this overhead is reasonable for the case of binary union, it is an important future challenge to reduce this overhead for larger unions.
  Finally, although \mcUCQAlg{} has the advantage of guaranteed delay (unlike that of \ouralgUCQ{} which is expected), our empirical evaluation shows that \ouralgUCQ{} is usually comparable to \mcUCQAlg{} or more efficient.

\section{Conclusions}\label{sec:conclusions}

We studied the problems of answering queries in a random
permutation and via a random-access.
We established that for CQs without self-joins it holds that
$\EN = \RA = \RP$. We also studied the generalization to unions of
free-connex CQs where, in contrast, we have $\EN \neq \RA$ and random-access may be intractable even if tractable for each CQ in the union.
We then studied two alternatives: (1) \mcUCQAlg uses the random-access approach for the restricted class of mc-UCQs and achieves
guaranteed $\log^2$ delay; (2) \ouralgUCQ finds a random permutation
directly for any UCQ comprising of free-connex CQs and achieves $\log$
delay in expectation. Our experimental study shows that the two
solutions are comparable on a union of two CQs, but
\ouralgUCQ{} preforms better on larger unions.
We described an
implementation of our algorithms, and presented an experimental study
showing that our algorithms outperform the sampling-with-rejection
alternatives.

It is an open problem to find which UCQs admit efficient fine-grained
enumeration, even without order guarantees~\cite{nofar:ucq}. However,
we do know that UCQs comprising of free-connex CQs admit efficient
enumeration. This work opens the question of finding an exact
characterisation for when such a UCQ admits random-access or random-permutation in polylogarithmic delay.

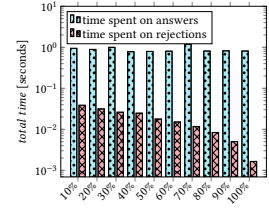
\begin{figure}[t]
    \centering
    \resizebox{0.4\linewidth}{!}{%
\begin{tikzpicture}[
  every axis/.style={
    ymin = 0,ymax = 10,
    symbolic x coords={10\%,20\%,30\%,40\%,50\%,60\%,70\%,80\%,90\%,100\%},
    xtick={10\%,20\%,30\%,40\%,50\%,60\%,70\%,80\%,90\%,100\%},
    bar width=6pt,
    ylabel near ticks,
    ylabel={$total\;time$ [seconds}],
    label style={font=\LARGE},
    tick label style={font=\LARGE},
    x tick label style={rotate=45,anchor=east}
  }
]

\begin{axis}[ybar,ymode=log,log origin=infty,legend pos=north west,legend style={font=\LARGE},legend cell align={left}]

\addplot+[black, fill=light_blue, postaction={pattern=crosshatch dots}] coordinates {(10\%,0.942738) (20\%,0.88764333333333) (30\%,0.99747833333333) (40\%,0.77979633333333) (50\%,0.79317866666667) (60\%,0.80292133333333) (70\%,1.1790766666667) (80\%,0.81006966666667) (90\%,0.820947) (100\%,0.81126766666667)};
\addlegendentry{time spent on answers}

\addplot+[black, fill=light_red, postaction={pattern=crosshatch}] coordinates {(10\%,0.038575) (20\%,0.03144) (30\%,0.026107666666667) (40\%,0.024753) (50\%,0.017811666666667) (60\%,0.015087666666667) (70\%,0.011743333333333) (80\%,0.0083836666666667) (90\%,0.0050213333333333) (100\%,0.0016493333333333)};
\addlegendentry{time spent on rejections}

\end{axis}

\end{tikzpicture}
}
    \caption{time spent on producing answers vs.~time spent on
      rejections across a full enumeration of $Q_7^S \cup Q_7^C$ (log
      scale).}
    \vskip-1em
    \label{fig:ucq1_rejection_vs_answers}
  \end{figure}

\bibliographystyle{abbrv}
\bibliography{references}

\clearpage
\appendix
\section*{APPENDIX}

\section{Proofs for Section~\ref{sec:ucqs}}

The description of the algorithm for the random access for $A\cup B$ gives the following lemma.

\begin{lemma}\label{lemma:AccessUnionOfTwo}
Let $A$ and $B$ be sets. Assume we have available enumeration algorithms for $A$, for $B$, and for 
$A\cap B$ such that 
\begin{enumerate}
 \item the enumeration order for $A\cap B$ is compatible with that for $A$
 \item after having carried out the preprocessing phase for $B$, 
  \begin{itemize}
   \item
     upon input of a number $j$ the 
     routine $B$.\Access($j$) returns, within time $t_B$, the $j$-th output element of the
     enumeration algorithm for $B$,
   \item
     upon input of an arbitrary $u$ it takes time $t_T$ to test if $u\in B$
  \end{itemize}
 \item after having carried out the preprocessing phase for $A$ we know its cardinality $|A|$,
   and upon input of a number $j$, the routine 
   $A$.\Access($j$) returns, within time $t_A$, the $j$-th output element of the
   enumeration algorithm for $A$
 \item after having carried out the preprocessing phase for $A\cap B$, 
   we know its cardinality $|A\cap B|$, and
   upon input of an arbitrary $c\in A\cap B$, its rank $(A\cap
   B).\Rank(c)$ according to the enumeration order for $A\cap B$ can be computed
   within time $t_{I}$
\end{enumerate}
Then, after having carried out the preprocessing phases for $A$, for $B$, and for $A\cap B$, Algorithm~\ref{alg:AccessUnionOfTwo} provides random-access to $A\cup B$ in such a way that 
upon input of an arbitrary number $j$ it
takes time $O(t_A+t_B+t_T+t_I)$ to output the $j$-th element enumerated by Algorithm~\ref{alg:DS-UnionTrick}.
\end{lemma}

The generalization to a union of an arbitrary number of sets is formzlied by the following
lemma.
\begin{lemma}\label{lemma:AccessUnionOfm}
Let $m\geq 2$ and let $S_1,\ldots,S_m$ be sets. For each $\ell\in [1,m]$ and each $I$ with $\emptyset\neq I\subseteq [\ell{+}1,m]$ let $T_{\ell,I}\deff S_{\ell}\cap\bigcap_{i\in I}S_i$.
Assume that for each $\ell\in[1,m]$ we have available an enumeration algorithm for $S_\ell$, and for each $\emptyset\neq I\subseteq [\ell{+}1,m]$ we have available an enumeration algorithm for 
$T_{\ell,I}$ such that
\begin{enumerate}
 \item the enumeration order for $T_{\ell,I}$ is compatible with that for $S_\ell$
 \item after having carried out the preprocessing phase for $S_\ell$ we know its cardinality $|S_\ell|$, and
   \begin{itemize}
   \item upon input of a number $j$ the routine $S_\ell.\Access(j)$
     returns, within time $t_{\textit{acc}}$ the $j$-th output element
     of the enumeration algorithm for $S_\ell$, and
     \item upon input of an arbitrary $u$ it takes time $t_{\textit{test}}$ to test if $u\in S_\ell$
   \end{itemize}
 \item\label{item:lemma:AccessUnionOfm:item3} after having carried out the preprocessing phase for $T_{\ell,I}$ we know its cardinality 
   $|T_{\ell,I}|$ and 
   \begin{itemize}
     \item upon input of an arbitrary $c\in T_{\ell,I}$ its rank $T_{\ell,I}$.\Rank($c$) can be computed within time $t_{\textit{inv-acc}}$, and
     \item upon input of an arbitrary $a\in S_\ell$ it takes time $t_{\textit{lar}}$ to return the particular $c\in T_{\ell,I}$ such that $c$ is the largest element of $T_{\ell,I}$ that is less than or equal to $a$ according to the enumeration order of $S_\ell$.
   \end{itemize}
\end{enumerate}
Then, after having carried out the preprocessing phases for $S_\ell$ and $T_{\ell,I}$ for all $\ell\in[1,m]$ and all $\emptyset\neq I\subseteq[\ell{+}1,m]$, we can provide random-access to 
$S_1\cup\cdots\cup S_m$ in such a way that upon input of an arbitrary number $j$ it takes time 
\[
  O(\, m {\cdot} t_{\textit{acc}} + m^2 {\cdot} t_{\textit{test}} + 2^m {\cdot} t_{\textit{inv-acc}} + 2^m {\cdot }t_{\textit{lar}} \,)
\]
to output the $j$-th element that is returned by the enumeration algorithm for $S_1\cup\cdots\cup S_m$ obtained by an iterated application of Algorithm~\ref{alg:DS-UnionTrick} (starting with $A=S_1$ and $B=S_2\cup\cdots\cup S_m$).
\end{lemma}
\begin{proof}
The proof follows the sketch described above. By applying the time bound obtained from Lemma~\ref{lemma:AccessUnionOfTwo} we obtain the following recursion for describing the time 
$f(m)$ used for providing access to the $j$-th element of the union of $m$ sets:
\[
  f(m) \ = \ 
  t_{\textit{acc}} + (m{-}1){\cdot}t_{\textit{test}}+ (2^{m-1}{-}1){\cdot}(t_{\textit{inv-acc}}+t_{\textit{lar}}) + f(m{-}1)
\]
Solving this recursion provides the claimed time bound.
\end{proof}

\begin{reptheorem}{\ref{thm:RAforUCQs}}
  \thmRAforUCQs
  \end{reptheorem}
\begin{proof}
Let $Q=Q_1\cup\cdots\cup Q_m$ be the given mc-UCQ, and 
let $\Algo{A}_I$, for all
$\emptyset\neq I\subseteq[1,m]$,
be $\RA$-algorithms which witness
that $Q$ is an mc-UCQ.

Upon input of a database $D$ we perform the linear-time preprocessing 
of all the algorithms $\Algo{A}_I$
input $D$.
Now consider an arbitrary $\ell\in[1,m]$ and an $I\subseteq[\ell{+}1,m]$.
For the sets $S_\ell\deff Q_{\set{\ell}}(D)$ and $T_{\ell,I}\deff
Q_{\set{\ell}\cup I}(D)$, we then immediately know that all the assumptions of
Lemma~\ref{lemma:AccessUnionOfm} are satisfied, except for the last
one (i.e., the last bullet point in
item~\eqref{item:lemma:AccessUnionOfm:item3} of
Lemma~\ref{lemma:AccessUnionOfm}).
Furthermore, we know that each of the time bounds $t_{\textit{test}}$,
$t_{\textit{acc}}$, and $t_{\textit{inv-acc}}$ are at most logarithmic
in the size $|D|$ of $D$. 
To finish the proof, it therefore suffices to show the following for
each $\ell\in [1,m]$ and each $\emptyset\neq I\subseteq [\ell{+}1,m]$:
\begin{itemize}
\item[($*$)]
On input of an arbitrary $a\in S_\ell$, within time $O(\log^2 |D|)$
we can output the
particular $c\in T_{\ell,I}$ such that $c$ is the largest element of
$T_{\ell,I}$ that is less than or equal to $a$ according to the
enumeration order of $S_\ell$.
\end{itemize}
We can achieve this by doing a binary search on indexes w.r.t.~the enumeration
orders on $S_\ell$ and $T_{\ell,I}$ by using the routines
$T_{\ell,I}$.\Access\  and $S_\ell$.\Rank.
More precisely, we start by letting $j=S_\ell$.\Rank($a$),
$c=T_{\ell,I}$.\Access(1), and $j_c=S_{\ell}.$\Rank($c$). 
If $j_c=j$ we
can safely return $c$. If $j_c>j$, we return an error message
indicating that $T_{\ell,I}$ does not contain any element less than or
equal to $a$. If $j_c<j$, we let $k_c=1$, $k_d=|T_{\ell,I}|$,
$d=T_{\ell,I}.\Access(k_d)$, and $j_d=S_\ell.\Rank(d)$.
If $j_d\leq j$ we can safely return $d$. Otherwise,
we do a binary search based on the invariant that $c,d$ are
elements of $T_{\ell,I}$ with $c<a<d$ (where $<$ refers to
the enumeration order of $S_\ell$), $j_c,j_d$ are their indexes in
$S_\ell$, and $k_c,k_d$ are their indexes in $T_{\ell,I}$:
we let $k'=\lfloor (k_c+k_d)/2 \rfloor$, and in case that $k'=k$ we
can safely terminate with output $c$. Otherwise, we let
$c'=T_{\ell,I}$.\Access($k'$) and $j'=S_\ell$.\Rank($c'$).
If $j'=j$ we can safely terminate and return $c'$. 
If $j'<j$ we proceed letting
$(c,d,j_c,j_d,k_c,k_d)=(c',d,j',j_d,k',k_d)$.
If $j'>j$ we proceed letting
$(c,d,j_c,j_d,k_c,k_d)=(c,c',j_c,j',k_c,k')$. 
The number of iterations is logarithmic in $|T_{\ell,I}|$, and each
iteration invokes a constant number of calls to $T_{\ell,I}$.\Access\ and
$S_{\ell}$.\Rank.
Since each such call is answered in time $O(\log|D|)$, we have
achieved $(*)$.
Theorem~\ref{thm:RAforUCQs} now follows from Lemma~\ref{lemma:AccessUnionOfm}.
\end{proof}

\section{Additions to Section~\ref{sec:experiments}} \label{sec:appendix-experiments}

\subsection{Queries} \label{sec:appendix_queries}
The following six queries are those used to compare \ouralgCQ{} to \zhaoalg{EW}.

\newcommand{\partitle}[1]{\vskip1em\par\noindent\underline{\textbf{#1}}\,}

\partitle{Query $Q_0$:} a chain join between the tables $\mathit{PARTSUPP}$, $\mathit{SUPPLIER}$, $\mathit{NATION}$, and $\mathit{REGION}$. It returns the suppliers that sell products (parts) along with their nation and region.
\begin{Verbatim}[frame=single]
SELECT DISTINCT r_regionkey, n_nationkey, 
                s_suppkey, ps_partkey 
  FROM region, nation, supplier, partsupp 
 WHERE r_regionkey = n_regionkey AND 
       n_nationkey = s_nationkey AND 
       s_suppkey = ps_suppkey
\end{Verbatim}

\partitle{Query $Q_2$:} similar to $Q_0$, except for the addition of the $\mathit{PART}$ table with \texttt{ps\_partkey = p\_partkey}.
\begin{Verbatim}[frame=single]
SELECT DISTINCT r_regionkey, n_nationkey,
                s_suppkey, ps_partkey
  FROM region, nation, supplier, partsupp, part
 WHERE r_regionkey = n_regionkey AND
       n_nationkey = s_nationkey AND
       s_suppkey = ps_suppkey AND
       ps_partkey = p_partkey
\end{Verbatim}

\partitle{Query $Q_3$:} the join of three tables: $\mathit{CUSTOMER}$, $\mathit{LINEITEMS}$, and $\mathit{ORDERS}$. 
We added the three attributes \texttt{l\_partkey}, \texttt{l\_suppkey}, and \texttt{l\_linenumber} to the output to ensure equivalence between set semantics and bag semantics.
\begin{Verbatim}[frame=single]
SELECT DISTINCT o_orderkey, c_custkey, l_partkey,
                l_suppkey, l_linenumber 
  FROM customer, orders, lineitems 
 WHERE c_custkey = o_custkey AND
       o_orderkey = l_orderkey;
\end{Verbatim}

\partitle{Query $Q_7$:} similar to $Q_3$, except it also joins $\mathit{SUPPLIER}$ and $\mathit{NATION}$ for the customer and the supplier.
\begin{Verbatim}[frame=single]
SELECT DISTINCT o_orderkey, c_custkey,
                n1.n_nationkey, s_suppkey, 
                l_partkey, l_linenumber, 
                n2.n_nationkey 
  FROM supplier, lineitem, orders, customer,
       nation n1, nation n2 
 WHERE s_suppkey = l_suppkey AND
       o_orderkey = l_orderkey AND
       c_custkey = o_custkey AND
       s_nationkey = n1.n_nationkey AND
       c_nationkey = n2.n_nationkey;
\end{Verbatim}
  
\partitle{Query $Q_9$:} the join of the tables $\mathit{NATION}$, $\mathit{SUPPLIER}$, $\mathit{LINEITEMS}$, $\mathit{PARTSUPP}$, $\mathit{ORDERS}$, and $\mathit{PART}$.
As in $Q_3$, we added the attributes \texttt{l\_partkey}, \texttt{l\_suppkey}, and \texttt{l\_linenumber} to the output to ensure an equivalence between bag and set semantics.
\begin{Verbatim}[frame=single]
SELECT DISTINCT n_nationkey, s_suppkey,
                o_orderkey, l_linenumber, p_partkey 
  FROM nation, supplier, lineitem,
       partsupp, orders, part 
 WHERE n_nationkey = s_nationkey AND
       s_suppkey = l_suppkey AND
       s_suppkey = ps_suppkey AND
       o_orderkey = l_orderkey AND
       l_partkey = p_partkey AND
       p_partkey = ps_partkey;
\end{Verbatim}

\partitle{Query $Q_{10}$:} similar to $Q_3$, except it also joins $\mathit{NATION}$.
\begin{Verbatim}[frame=single]
SELECT DISTINCT o_orderkey, c_custkey, l_partkey,
                l_suppkey, l_linenumber, n_nationkey  
  FROM lineitem, orders, customer, nation 
 WHERE o_orderkey = l_orderkey AND
       c_custkey = o_custkey AND
       c_nationkey = n_nationkey;
\end{Verbatim}

The UCQ experiments use $Q_7^S \cup Q_7^C$, $Q_2^N \cup Q_2^P \cup Q_2^S$, and $Q_A \cup Q_E$, with the following CQs:

\partitle{Query $Q_7^S$:} similar to $Q_7$, except for the addition of the constraint \texttt{n1.n\_name = "UNITED STATES"}. Meaning, the output should only include orders where the supplier is American.

\partitle{Query $Q_7^C$:} similar to $Q_7$, except we replace \texttt{n1.n\_name = "UNITED STATES"} with \texttt{n2.n\_name = "UNITED STATES"}. Meaning, demanding the customer is American (instead of the supplier being American).

\partitle{Query $Q_2^N$:} similar to $Q_2$, except for the addition of the constraint \texttt{n\_nationkey = 0}. Meaning, the supplier must be from the first country in the database.

\partitle{Query $Q_2^P$:} similar to $Q_2$, except for the addition of the constraint \texttt{n\_partkey mod 2 = 0}. Meaning, the part identifier must be even.

\partitle{Query $Q_2^S$:} similar to $Q_2$, except for the addition of the constraint \texttt{n\_suppkey mod 2 = 0}. Meaning, the supplier identifier must be even.

\partitle{Query $Q_A$:} the query deals with orders whose suppliers are from the United States of America. That is done by applying a condition to a full chain join of the tables $\mathit{ORDER}$, $\mathit{LINEITEM}$, $\mathit{SUPPLIER}$, $\mathit{NATION}$, and $\mathit{REGION}$.
\begin{Verbatim}[frame=single]
SELECT DISTINCT o_orderkey, s_suppkey,
                n_nationkey, r_regionkey, r_name 
  FROM orders, supplier, nation, region
 WHERE o_orderkey = l_orderkey AND
       l_suppkey = s_suppkey AND
       s_nationkey = n_nationkey AND
       n_regionkey = r_regionkey AND
       n_nationkey = 24
\end{Verbatim}

\partitle{Query $Q_E$:} similar to $Q_A$, except for the demand that the supplier be from the United Kingdom. Meaning, it has the same SQL expression as $Q_A$, but the constant 24 (United States) is replaced by 23 (United Kingdom).

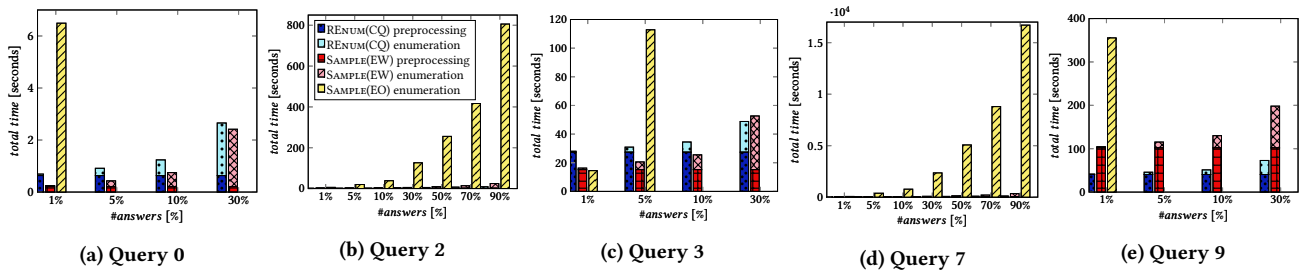
\begin{figure*}[htb]
\centering
    \begin{subfigure}{0.19\textwidth}
          \centering
          \resizebox{\linewidth}{!}{%
\begin{tikzpicture}[
  every axis/.style={ 
    ybar stacked,
    ymin = 0,ymax = 7,
    symbolic x coords={1\%, 5\%, 10\%, 30\%},
    xtick={1\%, 5\%, 10\%, 30\%},
    bar width=8pt,
    ylabel near ticks,
    xlabel near ticks,
    xlabel={$\#answers$ [\%]},
  ylabel={$total\;time$ [seconds}],
  label style={font=\LARGE},
  tick label style={font=\LARGE}
  }
]

\begin{axis}[bar shift=-15pt,hide axis]
\addplot+[black, fill=dark_blue, postaction={pattern=crosshatch dots}] coordinates
{(1\%,0.63365) (5\%,0.63365) (10\%,0.63365) (30\%,0.63365)}; %
\addplot+[black, fill=light_blue, postaction={pattern=dots}] coordinates
{(1\%,0.05733) (5\%,0.28) (10\%,0.603) (30\%,2.021667)}; %
\end{axis}

\begin{axis}[bar shift=-5pt]
\addplot+[black, fill=dark_red, postaction={pattern=grid}] coordinates
{(1\%,0.20233) (5\%,0.20233) (10\%,0.20233) (30\%,0.20233)}; %
\addplot+[black, fill=light_red, postaction={pattern=crosshatch}] coordinates
{(1\%,0.04733) (5\%,0.22833) (10\%,0.54533) (30\%, 2.21267)}; %
\end{axis}

\begin{axis}[bar shift=5pt]
\addplot+[black, fill=light_yellow, postaction={pattern=north east lines}] coordinates
{(1\%,6.49275) (5\%,0) (10\%,0) (30\%,0)}; %
\end{axis}

\end{tikzpicture}
}
          \caption{Query 0}
          \label{fig:q0_EO}
    \end{subfigure}
    \begin{subfigure}{0.19\textwidth}
          \centering
          \resizebox{\linewidth}{!}{%
\begin{tikzpicture}[
  every axis/.style={ 
    ybar stacked,
    ymin = 0,ymax = 850,
    symbolic x coords={1\%, 5\%, 10\%, 30\%, 50\%, 70\%, 90\%},
    xtick={1\%, 5\%, 10\%, 30\%, 50\%, 70\%, 90\%},
    bar width=8pt,
    ylabel near ticks,
    xlabel near ticks,
    xlabel={$\#answers$ [\%]},
  ylabel={$total\;time$ [seconds}],
  label style={font=\LARGE},
  tick label style={font=\LARGE}
  }
]

\begin{axis}[bar shift=-15pt,hide axis, legend pos=north west, legend style={font=\LARGE}, legend cell align={left}]
\addplot+[black, fill=dark_blue, postaction={pattern=crosshatch dots}] coordinates
{(1\%,0.8944565) (5\%,0.8944565) (10\%,0.8944565) (30\%,0.8944565) (50\%,0.8944565) (70\%,0.8944565) (90\%,0.8944565)};
\addlegendentry{\ouralgCQ{} preprocessing}

\addplot+[black, fill=light_blue, postaction={pattern=dots}] coordinates
{(1\%,0.076) (5\%,0.38733) (10\%,0.83167) (30\%,2.72433) (50\%,4.52567) (70\%,6.258) (90\%,8.05566)};
\addlegendentry{\ouralgCQ{} enumeration}

\addplot+[black, fill=dark_red, postaction={pattern=horizontal lines}] coordinates {(1\%,0)}; %
\addlegendentry{\zhaoalg{EW} preprocessing}

\addplot+[black, fill=light_red, postaction={pattern=north east lines}] coordinates {(1\%,0)}; %
\addlegendentry{\zhaoalg{EW} enumeration}

\addplot+[black, fill=light_yellow, postaction={pattern=grid}] coordinates {(1\%,0)}; %
\addlegendentry{\zhaoalg{EO} enumeration}
\end{axis}

\begin{axis}[bar shift=-5pt]
\addplot+[black, fill=dark_red, postaction={pattern=grid}] coordinates
{(1\%,1.75275) (5\%,1.75275) (10\%,1.75275) (30\%,1.75275) (50\%,1.75275) (70\%,1.75275) (90\%,1.75275)};
\addplot+[black, fill=light_red, postaction={pattern=crosshatch}] coordinates
{(1\%,0.092) (5\%,0.472) (10\%,1.069) (30\%,3.7493) (50\%,7.58233) (70\%,12.20233) (90\%,22.93566)};
\end{axis}

\begin{axis}[bar shift=5pt]
\addplot+[black, fill=light_yellow, postaction={pattern=north east lines}] coordinates
{(1\%,5.915) (5\%,19.949) (10\%,38.590) (30\%,126.349) (50\%,255.645) (70\%,416.728) (90\%,805.962)};
\end{axis}

\end{tikzpicture}
}
          \caption{Query 2}
          \label{fig:q2_EO}
    \end{subfigure}
    \begin{subfigure}{0.19\textwidth}
          \centering
          \resizebox{\linewidth}{!}{
\begin{tikzpicture}[
  every axis/.style={ 
    ybar stacked,
    ymin = 0,ymax = 120,
    symbolic x coords={1\%,5\%,10\%,30\%},
    xtick={1\%,5\%,10\%,30\%},
  bar width=8pt,
  ylabel near ticks,
  xlabel near ticks,
  xlabel={$\#answers$ [\%]},
  ylabel={$total\;time$ [seconds}],
  label style={font=\LARGE},
  tick label style={font=\LARGE}
  },
]

\begin{axis}[bar shift=-15pt,hide axis]
\addplot+[black, fill=dark_blue, postaction={pattern=crosshatch dots}] coordinates
{(1\%,27.35423) (5\%,27.35423) (10\%,27.35423) (30\%,27.35423)}; %
\addplot+[black, fill=light_blue, postaction={pattern=dots}] coordinates
{(1\%,0.675) (5\%,3.542) (10\%,7.134) (30\%,21.45467)}; %
\end{axis}

\begin{axis}[bar shift=-5pt]
\addplot+[black, fill=dark_red, postaction={pattern=grid}] coordinates
{(1\%,15.18867) (5\%,15.18867) (10\%,15.18867) (30\%,15.18867)}; %
\addplot+[black, fill=light_red, postaction={pattern=crosshatch}] coordinates
{(1\%,1.08033) (5\%,5.389) (10\%,10.34467) (30\%,37.47967)}; %
\end{axis}

\begin{axis}[bar shift=5pt]
\addplot+[black, fill=light_yellow, postaction={pattern=north east lines}] coordinates
{(1\%,14.493) (5\%,112.829) (10\%,0) (30\%,0)}; 
\end{axis}

\end{tikzpicture} 
}
          \caption{Query 3}
          \label{fig:q3_EO}
    \end{subfigure}
    \begin{subfigure}{0.19\textwidth}
          \centering
          \resizebox{\linewidth}{!}{%
\begin{tikzpicture}[
  every axis/.style={ 
    ybar stacked,
    ymin = 0,ymax = 17000,
    symbolic x coords={1\%,5\%,10\%,30\%,50\%,70\%,90\%}, 
    xtick={1\%,5\%,10\%,30\%,50\%,70\%,90\%},
  bar width=8pt,
  ylabel near ticks,
  xlabel near ticks,
  xlabel={$\#answers$ [\%]},
  ylabel={$total\;time$ [seconds}],
  label style={font=\LARGE},
  tick label style={font=\LARGE}
  },
]

\begin{axis}[bar shift=-15pt,hide axis]
\addplot+[black, fill=dark_blue, postaction={pattern=crosshatch dots}] coordinates
{(1\%,32.285) (5\%,32.285) (10\%,32.285) (30\%,32.285) (50\%,32.285) (70\%,32.285) (90\%,32.285)}; %
\addplot+[black, fill=light_blue, postaction={pattern=dots}] coordinates
{(1\%,0.88867) (5\%,4.805) (10\%,9.7533) (30\%,29.3493) (50\%,47.556) (70\%,67.004667) (90\%,85.63367)}; %
\end{axis}

\begin{axis}[bar shift=-5pt]
\addplot+[black, fill=dark_red, postaction={pattern=grid}] coordinates
{(1\%,38.23233) (5\%,38.23233) (10\%,38.23233) (30\%,38.23233) (50\%,38.23233) (70\%,38.23233) (90\%,38.23233)}; %
\addplot+[black, fill=light_red, postaction={pattern=crosshatch}] coordinates
{(1\%,1.48867) (5\%,7.38067) (10\%,14.332) (30\%,50.40167) (50\%,98.1633) (70\%,170.45567) (90\%,312.18067)}; %
\end{axis}

\begin{axis}[bar shift=5pt]
\addplot+[black, fill=light_yellow, postaction={pattern=north east lines}] coordinates
{(1\%,75.10833) (5\%,386.53567) (10\%,780.1755) (30\%,2360.315) (50\%,5074.519) (70\%,8795.437) (90\%,16704.456)}; %
\end{axis}

\end{tikzpicture} %
}
          \caption{Query 7}
          \label{fig:q7_EO}
    \end{subfigure}
    \begin{subfigure}{0.19\textwidth}
          \centering
          \resizebox{\linewidth}{!}{%
\begin{tikzpicture}[
  every axis/.style={ 
    ybar stacked,
    ymin = 0,ymax = 400,
    symbolic x coords={1\%,5\%,10\%,30\%}, 
    xtick={1\%,5\%,10\%,30\%},
  bar width=8pt,
  ylabel near ticks,
  xlabel near ticks,
  xlabel={$\#answers$ [\%]},
  ylabel={$total\;time$ [seconds}],
  label style={font=\LARGE},
  tick label style={font=\LARGE}
  },
]

\begin{axis}[bar shift=-15pt,hide axis]
\addplot+[black, fill=dark_blue, postaction={pattern=crosshatch dots}] coordinates
{(1\%,40.5219) (5\%,40.5219) (10\%,40.5219) (30\%,40.5219)}; %
\addplot+[black, fill=light_blue, postaction={pattern=dots}] coordinates
{(1\%,1.010667) (5\%,5.28833) (10\%,10.68567) (30\%,32.36533)}; %
\end{axis}

\begin{axis}[bar shift=-5pt]
\addplot+[black, fill=dark_red, postaction={pattern=grid}] coordinates
{(1\%,101.901) (5\%,101.901) (10\%,101.901) (30\%,101.901)}; %
\addplot+[black, fill=light_red, postaction={pattern=crosshatch}] coordinates
{(1\%,2.72066) (5\%,13.46667) (10\%,28.16267) (30\%,96.23133)}; %
\end{axis}

\begin{axis}[bar shift=5pt]
\addplot+[black, fill=light_yellow, postaction={pattern=north east lines}] coordinates
{(1\%,355.480) (5\%,0) (10\%,0) (30\%,0)}; %
\end{axis}

\end{tikzpicture} %
}
          \caption{Query 9}
          \label{fig:q9_EO}
    \end{subfigure}
     
    \caption{Total enumeration time of CQs when requesting different percentages of answers. In each bar, the bottom (darker)
part refers to the preprocessing phase and the top (lighter) part to the enumeration phase.}
    \label{fig:cq_total_time_EO}
  \end{figure*}

  {\begin{figure*}[b]
    \centering
    \begin{tabular}{|c|c|c|c|c|}\hline
     \textbf{algorithm} & \textbf{query} & \textbf{mean ($\mu$)} & \textbf{SD ($\sigma$)} & \textbf{outliers [\%]} \\
     \hline \hline
     \ouralgCQ{} & $Q_0$ & 3.964625 & 26.77761 & 2.3527 \\ \hline
     \zhaoalg{EW} & $Q_0$ & 3.965105 &  155.8571 & 6.6181 \\ \hline
     \ouralgCQ{} & $Q_2$ & 4.35985 & 29.14075 & 3.38665 \\ \hline
     \zhaoalg{EW} & $Q_2$ & 5.455966 & 136.3711 & 6.1348 \\ \hline
     \ouralgCQ{} & $Q_3$ & 4.443927 & 198.3520 & 3.07209 \\ \hline
     \zhaoalg{EW} & $Q_3$ & 6.599028 & 519.8907 & 6.17242 \\ \hline
     \ouralgCQ{} & $Q_7$ & 5.342141 & 191.1972 & 2.91181 \\ \hline
     \zhaoalg{EW} & $Q_7$ & 8.471535 & 534.6212 & 7.12741 \\ \hline
     \ouralgCQ{} & $Q_9$ & 5.664519 & 200.8911 & 2.8047 \\ \hline
     \zhaoalg{EW} & $Q_9$ & 14.90882 & 537.2356 & 8.26721 \\ \hline
     \ouralgCQ{} & $Q_{10}$ & 4.652109 & 195.7905 & 2.89414 \\ \hline
     \zhaoalg{EW} & $Q_{10}$ & 6.866843 & 519.3847 & 6.330277 \\ \hline
    \end{tabular}
    \quad\quad
    \begin{tabular}{|c|c|c|c|c|}\hline
     \textbf{algorithm} & \textbf{query} & \textbf{mean ($\mu$)} & \textbf{SD ($\sigma$)} & \textbf{outliers [\%]} \\
     \hline \hline
     \ouralgCQ{} & $Q_0$ & 3.891264 & 18.9752 & 2.685475 \\ \hline
     \zhaoalg{EW} & $Q_0$ & 20.28385 &  2848.952 & 12.0188 \\ \hline
     \ouralgCQ{} & $Q_2$ & 4.319361 & 20.74831 & 3.662525 \\ \hline
     \zhaoalg{EW} & $Q_2$ & 38.02335 & 4815.667 & 12.412425 \\ \hline
     \ouralgCQ{} & $Q_3$ & 4.347431 & 140.2731 & 3.65655 \\ \hline
     \zhaoalg{EW} & $Q_3$ & 49.84528 & 20863.52 & 12.3539\\ \hline
     \ouralgCQ{} & $Q_7$ & 5.254392 & 135.2184 & 3.47616\\ \hline
     \zhaoalg{EW} & $Q_7$ &  72.04367 & 30804.44 & 12.50615 \\ \hline
     \ouralgCQ{} & $Q_9$ & 5.57028 & 142.0814 & 3.2938 \\ \hline
     \zhaoalg{EW} & $Q_9$ & 141.239102 & 56781.80 & 12.75 \\ \hline
     \ouralgCQ{} & $Q_{10}$ & 4.564015 & 138.4678 & 3.43488 \\ \hline
     \zhaoalg{EW} & $Q_{10}$ & 49.30842 & 11648.17 & 12.4268 \\ \hline
    \end{tabular}
    \caption{The mean and standard deviation of the delay during
      \e{(a)} an
      enumeration of  50\% of answers using each algorithm (left), \e{and (b)}
                a full enumeration using each algorithm (right).
               }
    \label{fig:mean_and_dev_tables}
\end{figure*}}

\subsection{Additional methods by Zhao et al.}\label{sec:other_methods}

As mentioned in \Cref{sec:algs}, Zhao et al.~\cite{zhao2018random} discuss 4 different ways of initializing their sampling algorithm, denoted as \e{RS}, \e{EO}, \e{OE}, and \e{EW}. In the implementation of Zhao et al.~, \e{EW} and \e{EO} were implemented for every query. In addition, there is an implementation of RS and OE for $Q_3$. Here we review \e{EO}, \e{OE}, and \e{RS} (in sections \ref{sec:EO}, \ref{sec:OE}, and \ref{sec:RS} respectively) in order to explain our comparison to \zhaoalg{EW} alone.

\subsubsection{\e{EO}} \label{sec:EO}
As \zhaoalg{EO} may reject, it possesses a much longer sampling time (as evident by our experiments). 
\Cref{fig:cq_total_time_EO} repeats the experiment made in section~\ref{sec:cq_total_time} (depicted in \Cref{fig:cq_total_time}) with the addition of \zhaoalg{EO}.
We omit the \zhaoalg{EO} preprocessing, as Zhao et al.~\cite{zhao2018random} did in their work, and as it underperforms compared to \zhaoalg{EW} regardless.
When running \zhaoalg{EO}, we used a timeout and halted when it took longer than 100 times the sampling time of its \e{EW} counterpart.
When \zhaoalg{EO} timed-out, the corresponding bar in \Cref{fig:cq_total_time_EO} is omitted. In addition, we omit $Q_10$, as \zhaoalg{EO} did not produce 1\% of the answers within the time limit.
\Cref{fig:cq_total_time_EO} shows that with the exception of $Q_3$ at 1\%, \e{EO} is significantly slower than both \ouralgCQ{} and \zhaoalg{EW}.

\subsubsection{\e{OE}} \label{sec:OE}
 Out of our six queries, \zhaoalg{OE} was implemented on $Q_3$ alone. \Cref{fig:cq_total_time_OE} shows the results of section \ref{sec:cq_total_time} with \zhaoalg{OE} added. In our experiments with $Q_3$, \zhaoalg{EW} has always out-preformed \zhaoalg{OE}.

\begin{figure}[t]
\centering
\resizebox{0.6\linewidth}{!}{%
\begin{tikzpicture}[
  every axis/.style={ 
    ybar stacked,
    ymin = 0,ymax = 350,
    symbolic x coords={1\%,5\%,10\%,30\%,50\%,70\%,90\%},
    xtick={1\%,5\%,10\%,30\%,50\%,70\%,90\%},
  bar width=5pt,
  ylabel near ticks,
  xlabel near ticks,
  xlabel={$\#answers$ [\%]},
  ylabel={$total\;time$ [seconds}],
  label style={font=\LARGE},
  tick label style={font=\LARGE}
  },
]

\begin{axis}[bar shift=-10pt,hide axis,legend pos=north west,legend style={font=\Large},legend cell align={left}]
\addplot+[black, fill=dark_blue, postaction={pattern=crosshatch dots}] coordinates
{(1\%,27.35423) (5\%,27.35423) (10\%,27.35423) (30\%,27.35423) (50\%,27.35423) (70\%,27.35423) (90\%,27.35423)}; %
\addlegendentry{\ouralgCQ{} preprocessing}

\addplot+[black, fill=light_blue, postaction={pattern=dots}] coordinates
{(1\%,0.675) (5\%,3.542) (10\%,7.134) (30\%,21.45467) (50\%,34.089) (70\%,48.14733) (90\%,61.551)}; %
\addlegendentry{\ouralgCQ{} enumeration}

\addplot+[black, fill=dark_red, postaction={pattern=horizontal lines}] coordinates
{(1\%,0)}; %
\addlegendentry{\zhaoalg{EW} preprocessing}

\addplot+[black, fill=light_red, postaction={pattern=north east lines}] coordinates
{(1\%,0)}; %
\addlegendentry{\zhaoalg{EW} enumeration}

\addplot+[black, fill=dark_green, postaction={pattern=vertical lines}] coordinates
{(1\%,0)}; %
\addlegendentry{\zhaoalg{OE} preprocessing}

\addplot+[black, fill=light_green, postaction={pattern=north west lines}] coordinates
{(1\%,0)}; %
\addlegendentry{\zhaoalg{OE} enumeration}
\end{axis}

\begin{axis}[bar shift=-5pt]
\addplot+[black, fill=dark_red, postaction={pattern=horizontal lines}] coordinates
{(1\%,15.18867) (5\%,15.18867) (10\%,15.18867) (30\%,15.18867) (50\%,15.18867) (70\%,15.18867) (90\%,15.18867)}; %
\addplot+[black, fill=light_red, postaction={pattern=north east lines}] coordinates
{(1\%,1.08033) (5\%,5.389) (10\%,10.34467) (30\%,37.47967) (50\%,70.85967) (70\%,122.64067) (90\%,221.76033)}; %
\end{axis}

\begin{axis}[bar shift=0pt]
\addplot+[black, fill=dark_green, postaction={pattern=vertical lines}] coordinates
{(1\%,8.3000766666667) (5\%,8.1846266666667) (10\%,8.1564066666667) (30\%,12.923033333333) (50\%,13.096233333333) (70\%,13.1976) (90\%,13.249233333333)}; %

\addplot+[black, fill=light_green, postaction={pattern=north west lines}] coordinates
{(1\%,7.9473766666667) (5\%,19.453) (10\%,27.409733333333) (30\%,63.0288) (50\%,107.13733333333) (70\%,179.19466666667) (90\%,327.37933333333)}; %
\end{axis}

\end{tikzpicture} %
}
\caption{Total enumeration time of $Q_3$ when requesting different percentages of answers. In each bar, the bottom (darker)
part refers to the preprocessing phase and the top (lighter) part to the enumeration phase.}
\label{fig:cq_total_time_OE}
\end{figure}
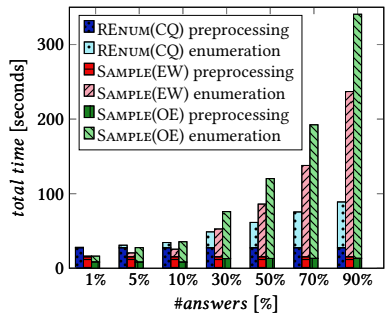
 
\subsubsection{\e{RS}} \label{sec:RS}
\zhaoalg{RS} was also implemented only on $Q_3$.
\zhaoalg{RS} was unable to produce a sample of 1\% of the answers to $Q_3$ in less than an hour. It took \zhaoalg{RS} about 6.8 seconds to gather a sample of 100000 distinct answers, which is roughly 0.33\% of all answers. Therefore, \zhaoalg{RS} would be slower than \zhaoalg{EW} even if it were to proceed and sample 1\% with no deterioration due to repeating samples.

\subsection{More on the CQ delay experiment} \label{sec:delay_tables}

This section describes further information that was left out of
section~\ref{sec:cq_delay} to save space. The following tables show
the mean, standard deviation (SD) and number of delay samples that
counted as outliers during the enumeration of 50\% of the answers (on
the left of \Cref{fig:mean_and_dev_tables}) and a full enumeration (on
the right of \Cref{fig:mean_and_dev_tables}).  We see that \ouralgCQ{}
always possesses a smaller mean than \zhaoalg{EW}, sometimes by an
order of magnitude. We also see that \ouralgCQ{} always has
considerably lower standard deviation than that of \zhaoalg{EW}. That
holds even when the two are close in median as is the case with
$Q_0$. Finally, the number of outliers in \ouralgCQ{} boxplots is also
consistently smaller. The smaller number of outliers and lower
standard deviation indicates the predictability of the delay, as it
does not grow rapidly.

\end{document}